\def\full{}

\ifdefined\full
\documentclass[runningheads]{llncs}
\makeatletter
\newcommand\footnoteref[1]{\protected@xdef\@thefnmark{\ref{#1}}\@footnotemark}
\makeatother

\else
    \ifdefined\llncs
    \documentclass{llncs}
    \usepackage{orcidlink}
    \else
    \documentclass[]{iacrcc}
    \PassOptionsToPackage{bookmarksopen=true,bookmarksopenlevel=2}{hyperref}
    \PassOptionsToPackage{hyphens}{url}
    \fi
\fi

\usepackage{a4wide}

\usepackage[hyphens]{url}
\usepackage{tablefootnote}
\usepackage[T1]{fontenc}
\usepackage[flushleft]{threeparttable}
\usepackage{amssymb,amsfonts}
\usepackage{tabularx}
\usepackage{multirow}
\usepackage{arydshln}
\usepackage{cryptocode}
\usepackage{xspace}
\usepackage{bm}
\usepackage{booktabs}
\usepackage{float}
\usepackage[inline]{enumitem}
\usepackage{pgfplots}
\usepackage{pgfplotstable}
\usepackage{csquotes}
\usepackage{csvsimple}
\usepackage{color}
\usepackage{hyperref}
\usepackage{cleveref}

\definecolor{darkgreen}{rgb}{0,0.4,0}

\pgfplotsset{compat=1.16}
\pgfplotstableread{
pqsig                         Kyber1024    NTRU-HPS-2048-509
SPHINCS+-SHA256-256f-robust   1.054547026  1.055271444
SPHINCS+-SHA256-256f-simple   0.665593256  0.65087234
SPHINCS+-SHA256-256s-robust   6.915164492  6.906950926
SPHINCS+-SHA256-256s-simple   3.197696194  3.22271654
SPHINCS+-SHAKE256-256f-robust 1.322368788  1.351088144
SPHINCS+-SHAKE256-256f-simple 0.863275632  0.905718094
SPHINCS+-SHAKE256-256s-robust 8.877232394  9.009868012
SPHINCS+-SHAKE256-256s-simple 5.075884542  5.083704802
Dilithium5                    0.337269602  0.344824372
Dilithium5-AES                0.339753544  0.347413216
Falcon-1024                   0.371352454  0.383805688
Picnic-L5-FS                  0.415450074  0.42177718
Picnic-L5-full                0.378677316  0.3742238
Picnic3-L5                    0.51773609   0.53683842
}{\tableone}
\pgfplotstablegetcolsof{\tableone}
\pgfmathtruncatemacro{\NoOfColsOne}{\pgfplotsretval-1}

\pgfplotsset{compat=1.16}
\pgfplotstableread{
pqsig                         Kyber1024    NTRU-HPS-2048-509
SPHINCS+-SHA256-256f-robust   1.846034756  1.881519946
SPHINCS+-SHA256-256f-simple   1.398056834  1.504525442
SPHINCS+-SHA256-256s-robust   7.674954504  7.658467406
SPHINCS+-SHA256-256s-simple   4.005654522  3.936769616
SPHINCS+-SHAKE256-256f-robust 2.120464742  2.10853786
SPHINCS+-SHAKE256-256f-simple 1.65667872   1.683489464
SPHINCS+-SHAKE256-256s-robust 9.644953748  9.85840589
SPHINCS+-SHAKE256-256s-simple 5.839800348  5.921163802
Dilithium5                    1.118972784  1.109398124
Dilithium5-AES                1.119181812  1.095479176
Falcon-1024                   1.161695532  1.151458768
Picnic-L5-FS                  1.218797276  1.172911332
Picnic-L5-full                1.15038813   1.190677128
Picnic3-L5                    1.271207926  1.325721362
}{\tabletwo}
\pgfplotstablegetcolsof{\tabletwo}
\pgfmathtruncatemacro{\NoOfColsTwo}{\pgfplotsretval-1}

\usetikzlibrary{matrix,positioning,chains,shapes.geometric, arrows, fit}
\tikzstyle{QKDmodule} = [rectangle, minimum width=1cm, minimum height=0.5cm, text centered, font=\scriptsize, draw=black]
\tikzstyle{KM} = [rectangle, minimum width=2cm, minimum height=0.5cm, text centered, font=\scriptsize, draw=black]
\tikzstyle{QKDN} = [rectangle, minimum width=10cm, minimum height=0.5cm, text centered, font=\scriptsize, draw=black]
\tikzstyle{arrowtwo} = [<->,>=stealth]
\tikzstyle{arrow} = [->,>=stealth]
\tikzstyle{node} = [rectangle, minimum width=2.2cm, minimum height=0.5cm, text centered, font=\scriptsize, draw=black!60, dashed]
\tikzstyle{app} = [rectangle, minimum width=2cm, minimum height=1cm, text centered, font=\scriptsize, draw=black]
\tikzstyle{desc} = [text centered, font=\scriptsize]
\tikzstyle{layer} = [rectangle ,minimum width=2cm, minimum height=0.5cm, align=right, font=\scriptsize]

\floatstyle{plain}
\newfloat{experiment}{ht}{exps}
\floatname{experiment}{Experiment}

\crefname{experiment}{experiment}{experiments}
\Crefname{experiment}{Experiment}{Experiments}

\newcommand{\nist}{{\tt NIST}}
\newcommand{\qkd}{{\tt QKD}}
\newcommand{\pqc}{{\tt PQC}}

\newcommand{\INDCPA}{\ensuremath{\mathsf{ind}\text{-}\allowbreak\mathsf{cpa}}}
\newcommand{\INDCCA}{\ensuremath{\mathsf{ind}\text{-}\allowbreak\mathsf{cca}}}
\newcommand{\INDONECCA}{\ensuremath{\mathsf{ind}\text{-}\allowbreak\mathsf{1cca}}}
\newcommand{\EUFCMA}{\ensuremath{\mathsf{EUF}\text{-}\allowbreak\mathsf{CMA}}}

\newcommand{\AdvEUFCMA}[1]{\ensuremath{\mathsf{Adv}^{\mathsf{euf\text{-}cma}}_{#1}}}
\newcommand{\adv}{\advA}
\newcommand{\Adv}{\mathsf{Adv}}

\newcommand{\mathcmd}[1]{\ensuremath{#1}\xspace} 
\newcommand{\N}{\mathcmd{\mathbb{N}}}
\newcommand{\setS}{\mathcmd{\mathcal{S}}}
\newcommand{\advA}{\ensuremath{\mathcal{A}}\xspace}
\newcommand{\advB}{\ensuremath{\mathcal{B}}\xspace}
\newcommand{\advC}{\ensuremath{\mathcal{C}}\xspace}
\newcommand{\advD}{\ensuremath{\mathcal{D}}\xspace}
\newcommand{\advE}{\ensuremath{\mathcal{E}}\xspace}

\newcommand{\getsr}{\xleftarrow{R}}
\newcommand{\pk}{\mathcmd{\mathsf{pk}}} 
\newcommand{\sk}{\mathcmd{\mathsf{sk}}} 
\newcommand{\kk}{\secpar} 
\newcommand{\msg}{\mathcmd{m}} 
\newcommand{\Msg}{\mathcmd{\mathcal{M}}} 

\newcommand{\key}{\mathcmd{K}} 
\newcommand{\ctxt}{\mathcmd{c}} 
\newcommand{\ExpEUFCMA}[1]{\ensuremath{\mathsf{Exp}^{\mathsf{euf\text{-}cma}}_{#1}}}
\newcommand{\prfeval}{\ensuremath{\mathcal{F}}}

\renewcommand{\secpar}{\kappa}
\newcommand{\kgen}{{\mathsf{KGen}}}
\newcommand{\GetKey}{{\mathsf{GetKey}}}
\newcommand{\ssign}{\mathsf{Sign}}
\newcommand{\sverify}{\mathsf{Ver}}
\newcommand{\Key}{\mathcmd{\mathcal{K}}}
\newcommand{\KEM}{\mathcmd{\mathsf{KEM}}} 

\newcommand{\KEMINDCPA}{IND-CPA\xspace}
\newcommand{\KEMINDCCA}{IND-CCA\xspace}
\newcommand{\ExpKEMINDT}[1]{\ensuremath{\mathsf{Exp}^{\mathsf{ind\text{-}}T}_{#1}}}
\newcommand{\ExpKEMINDCPA}[1]{\ensuremath{\mathsf{Exp}^{\mathsf{ind\text{-}cpa}}_{#1}}}
\newcommand{\AdvKEMINDCPA}[1]{\ensuremath{\mathsf{Adv}^{\mathsf{ind\text{-}cpa}}_{#1}}}
\newcommand{\ExpKEMINDCCA}[1]{\ensuremath{\mathsf{Exp}^{\mathsf{ind\text{-}cca}}_{#1}}}
\newcommand{\AdvKEMINDCCA}[1]{\ensuremath{\mathsf{Adv}^{\mathsf{ind\text{-}cca}}_{#1}}}
\newcommand{\Encaps}{\mathcmd{\mathsf{Enc}}} 
\newcommand{\Decaps}{\mathcmd{\mathsf{Dec}}} 

\newcommand{\MAC}{\mathcmd{\mathsf{MAC}}}
\newcommand{\KDF}{\mathcmd{\mathsf{KDF}}}
\newcommand{\msign}{\mathsf{Auth}}
\newcommand{\mverify}{\mathsf{Ver}}

\newcommand{\dPRF}{\mathcmd{\mathsf{dual-PRF}}}

\newcommand{\IHTS}{\mathcmd{\mathsf{IHTS}}}
\newcommand{\RHTS}{\mathcmd{\mathsf{RHTS}}}
\newcommand{\dHS}{\mathcmd{\mathsf{dHS}}}
\newcommand{\AHS}{\mathcmd{\mathsf{AHS}}}
\newcommand{\IAHTS}{\mathcmd{\mathsf{IAHTS}}}
\newcommand{\RAHTS}{\mathcmd{\mathsf{RAHTS}}}
\newcommand{\dAHS}{\mathcmd{\mathsf{dAHS}}}
\newcommand{\MS}{\mathcmd{\mathsf{MS}}}

\newcommand{\IF}{\mathcmd{\mathsf{IF}}}
\newcommand{\RF}{\mathcmd{\mathsf{RF}}}
\newcommand{\IATS}{\mathcmd{\mathsf{IATS}}}
\newcommand{\RATS}{\mathcmd{\mathsf{RATS}}}
\newcommand{\fk}{\mathcmd{\mathsf{fk}}}

\newcommand{\Muckle}{\mathcmd{\text{Muckle}}}

\newcommand{\SigMuckle}{\mathcmd{\text{Muckle}+}}
\newcommand{\KEMuckle}{\mathcmd{\text{Muckle}\#}}
\newcommand{\MuckleSharp}{\KEMuckle}
\newcommand{\lbl}{\mathcmd{\ell}}
\newcommand{\SecState}{\mathcmd{\mathsf{SecState}}}

\newcommand{\KEMTLS}{\mathcmd{\text{KEMTLS}}}

\newcommand{\clean}{\mathcmd{\mathsf{clean}}}
\newcommand{\Corrupt}{\mathcmd{\mathsf{Corrupt}}}

\newcommand{\CorruptSK}{\mathcmd{\mathsf{CorruptSK}}}
\newcommand{\CorruptQK}{\mathcmd{\mathsf{CorruptQK}}}
\newcommand{\CorruptCK}{\mathcmd{\mathsf{CorruptCK}}}
\newcommand{\Compromise}{\mathcmd{\mathsf{Compromise}}}
\newcommand{\CompromiseSS}{\mathcmd{\mathsf{CompromiseSS}}}
\newcommand{\CompromiseSK}{\mathcmd{\mathsf{CompromiseSK}}}
\newcommand{\CompromiseQK}{\mathcmd{\mathsf{CompromiseQK}}}
\newcommand{\CompromiseCK}{\mathcmd{\mathsf{CompromiseCK}}}
\newcommand{\Test}{\mathcmd{\mathsf{Test}}}
\newcommand{\Create}{\mathcmd{\mathsf{Create}}}
\newcommand{\Reveal}{\mathcmd{\mathsf{Reveal}}}
\newcommand{\Send}{\mathcmd{\mathsf{Send}}}
\newcommand{\act}{\mathcmd{\mathsf{active}}}
\newcommand{\acc}{\mathcmd{\mathsf{accept}}}
\newcommand{\rej}{\mathcmd{\mathsf{reject}}}

\title{Quantum-Safe Hybrid Key Exchanges with KEM-Based Authentication}

\ifdefined\full
\author{Christopher Battarbee\inst{1} \and Christoph Striecks\inst{2} \and Ludovic Perret\inst{1,4} \and Sebastian Ramacher\inst{2} \and Kevin Verhaeghe\inst{3}\thanks{Work done while with AIT.}}

\institute{Sorbonne University, CNRS, LIP6
F-75005 Paris, France\email{} \and AIT Austrian Institute of Technology, Vienna, Austria\email{} \and ETH Zurich, Zurich, Switzerland\email{}
\and 
Laboratoire de Recherche de l’EPITA,
Le Kremlin-Bicêtre, France
}
\fi

\begin{document}

\maketitle

\begin{abstract} Authenticated Key Exchange (AKE) between any two entities is one of the most important security protocols available for securing our digital networks and infrastructures. In PQCrypto 2023, Bruckner, Ramacher and Striecks proposed a novel hybrid AKE (HAKE) protocol dubbed $\SigMuckle$ that is particularly useful in large quantum-safe networks consisting of a large number of nodes. Their protocol is hybrid in the sense that it allows key material from conventional, post-quantum, and quantum cryptography primitives to be incorporated into a single end-to-end authenticated shared key.

To achieve the desired authentication properties, $\SigMuckle$ utilizes post-quantum digital signatures. However, available instantiations of such signatures schemes are not yet efficient enough compared to their post-quantum key-encapsulation mechanism (KEM) counterparts, particularly in large networks with potentially several connections in a short period of time.

To mitigate this gap, we propose $\MuckleSharp$ that pushes the efficiency boundaries of currently known HAKE constructions. $\MuckleSharp$ uses post-quantum key-encapsulating mechanisms for implicit authentication inspired by recent works done in the area of Transport Layer Security (TLS) protocols, particularly, in KEMTLS (CCS'20). 

We port those ideas to the HAKE framework and develop novel proof techniques on the way. Due to our KEM-based approach, the resulting protocol has a slightly different message flow compared to prior work that we carefully align with the HAKE framework and which makes our changes to $\SigMuckle$ non-trivial.

\ifdefined\full
\else
Lastly, we evaluate the approach by a prototypical implementation and a direct comparison with $\SigMuckle$ to highlight the efficiency gains.
\fi

\end{abstract}

\section{Introduction}

The continuous progress of quantum technologies is making the deployment of  global quantum-safe infrastructures an increasingly pressing matter. Perhaps most visibly in this direction, the first post-quantum cryptographic standards were formally published by the National Institute of Standards and Technology (\nist{})\footnote{\url{https://csrc.nist.gov/projects/post-quantum-cryptography}} in $2024$, following a multi-year and massive effort from a large community of researchers. Indeed, these quantum-resistant, or ``post-quantum'', technologies are entering a stage of maturity whereby their deployment into existing network architectures is a growing field of study.

In practice, the new post-quantum standards are expected to be deployed in tandem with existing techniques from the field of \textit{Authenticated} Key Exchanges (AKEs)~\cite{C:Maurer92,STOC:BelRog95}. In a public network of peers, for example, all the security guarantees of a post-quantum key encapsulation mechanism ($\KEM$) are void if one party can impersonate the other, so we seek to guarantee that any two parties in the network can verify each other's authenticity -- such a guarantee is called end-to-end authenticity. In an AKE protocol, we seek to establish a session key while achieving end-to-end authenticity.

\paragraph{Motivation.} There are several techniques by which one can achieve end-to-end authenticity, including using some pre-shared keying material (PSK). However, it is a well-known result that in a network of $n$ peers one requires $\mathcal{O}(n^2)$ initial PSK distributions, so this type of authentication might not be suitable in a large or dynamic network scenario, particularly if we wish to be able to add new peers to the network.

In such a scenario, it may be more appropriate to use certificate-based authentication. Here, we suppose that the authentication is based on asymmetric cryptography, and some third-party, trusted certificate authority (CA) guarantees the authenticity of each peer's public keys. In turn, these asymmetric cryptography protocols are used, in some capacity, during an AKE protocol, establishing the authenticity of each party. 

Typically, a digital signature scheme is utilized. One of the more famous examples of an AKE protocol is the Transport Layer Security (TLS) 1.3~\cite{DBLP:journals/rfc/rfc8446} protocol where two parties run an  unauthenticated ephemeral Diffie-Hellman-style handshake, and then sign various parts of the transcript, yielding authenticity. Several projects and initiatives are working on migrating network protocols such as TLS 1.3 to the post-quantum setting, both for KEMs and signatures, and evaluating their efficiency.\footnote{E.g., \url{https://blog.cloudflare.com/pq-2024/}, \url{https://www.microsoft.com/en-us/research/project/post-quantum-tls/}, \url{https://security.googleblog.com/2024/08/post-quantum-cryptography-standards.html}} The currently available post-quantum standards~\cite{NISTPQC:CRYSTALS-KYBER22,NISTPQC:CRYSTALS-DILITHIUM22,NISTPQC:FALCON22,NISTPQC:SPHINCS+22} are such that the $\KEM$s, in contrast to the pre-quantum setting, are more efficient than the currently available post-quantum signatures. In the post-quantum setting, therefore, it is tempting to also consider using $\KEM$s as an authentication mechanism; particularly following the seminal approach taken by Schwabe, Stebila, and Wiggers on $\KEMTLS$ in \cite{CCS:SchSteWig20} (which itself is based on the OPTLS proposal by Krawczyk and Wee \cite{7467348}).

\paragraph{Towards defense-in-depth approaches.} In parallel to the post-quantum standardization effort, the coming quantum threat has also motivated the deployment of large \qkd{}-based testbeds; notably in China\footnote{\url{https://physicsworld.com/a/quantum-cryptography-network-spans-4600-km-in-china/}}, in the EU with the EuroQCI network\footnote{\url{https://digital-strategy.ec.europa.eu/en/policies/european-quantum-communication-infrastructure-euroqci}} and the various testbeds distributed throughout Europe~\cite{DBLP:conf/icton/RaddoRLOM19,DBLP:conf/icton/MartinBOBVSSACSSEDRPL23,DBLP:journals/entropy/BrauerVBMBGRBFPPLMB24}, and in the UK with the Quantum Hub\footnote{\url{https://uknqt.ukri.org/success-stories/uk-quantum-networks/}}.
Whilst the first motivation for such testbeds was arguably to increase the maturity of \qkd{} devices, the next step is to move towards a practical \qkd{}-based network in real-life use-cases. This leads to the problem of integrating \qkd{} with existing security protocols, and to the study of the combination (or, hybridization) of \qkd{} with modern cryptography. 
Particularly, the European Commission recommends \pqc{} in hybridization with currently deployed cryptographic primitives or \qkd{}.\footnote{\url{https://digital-strategy.ec.europa.eu/en/news/commission-publishes-recommendation-post-quantum-cryptography}} 

Recently, hybrid AKE (HAKE) \cite{PQCRYPTO:DowHanPat20,PQCRYPTO:BruRamStr23} has emerged as a so-called ``defense-in-depth'' solution to these problems, incorporating key material from conventional and \pqc{} primitives, as well as from \qkd{}. The overall goal is resilience: despite the intensive research into their security, the post-quantum standards are young algorithms whose failure is not unprecedented; but the promised unconditional security of $\qkd{}$ is caveated by the relatively immature devices on which current $\qkd{}$  is implemented. One aims, therefore, to utilize multiple sources of key material, in such a way that if all but one of these key-material sources fail, the resulting shared key between any two parties is still authentic and confidential.  

In a different context, $\KEM$ combiners \cite{PKC:GiaHeuPoe18} were proposed to enhance the security of combined key material from different $\KEM$s on the primitive level. On the protocol level, we want to achieve more in terms of authentication and forward security which makes HAKE approaches the most suitable ones.

\paragraph{Forward security.} As well as tolerating ``real-time'' faults, protocols in the HAKE setting naturally have a higher degree of forward security, in the following sense. 
Forward security is an essential security features in nowadays protocols as it guarantees the security of past protocol sessions even in case of key leakage in the current session. 
Such features has been demonstrated to be of significant interest in interactive key-exchange protocols~\cite{EC:Gunther89a,DBLP:journals/dcc/DiffieOW92,CANS:DDGHJKNRW20,EC:RosSlaStr23}, public-key encryption \cite{EC:CanHalKat03,EPRINT:Groth21}, digital signatures \cite{C:BelMin99,USENIX:DGNW20}, search on encrypted data \cite{CCS:BosMinOhr17}, 0-RTT key exchange \cite{EC:GHJL17,EC:DJSS18,AC:CRSS20,DBLP:journals/joc/DerlerGJSS21}, updatable cryptography \cite{TCC:SlaStr23}, mobile Cloud backups \cite{DBLP:conf/osdi/DautermanCM20}, proxy cryptography \cite{PKC:DKLRSS18}, Tor \cite{PoPETS:LGMHS20}, and content-delivery networks \cite{FC:DRSS21}, among others.

In our AKE setting, consider the security model in which the security proof of TLS 1.3 is executed \cite{JC:DFGS21}; one does not model for the scenario in which the ephemeral secret established by a Diffie-Hellman~\cite{DifHel76} handshake is compromised after the session has accepted. Of course, such a compromise would completely reveal all secrets derived by that session of TLS 1.3, so there is a sense in which the secret established by the ephemeral Diffie-Hellman handshake has to remain secret ``forever.'' Obviously this is not realistic in a post-quantum setting; the problem can in part be handled by updating authenticated key-exchange mechanisms to use post-quantum ephemeral handshakes (probably via a post-quantum $\KEM$). On the other hand, a scenario in which all confidence is rapidly lost in a highly-regarded $\KEM$ candidate is not unprecedented, as seen with the SIDH break~\cite{EC:CasDec23,EC:Robert23,maino2023direct}. 

Switching to the HAKE model improves forward security here, because as well as mixing three sources of secrets in such a way that any two can fail at the same time, the \qkd{} link is resistant to the store-now-decrypt-later attacks relevant to both classical and post-quantum $\KEM$s used for the ephemeral confidentiality part (whereby assuming of course working mechanisms for end-to-end authentication). This is essentially a consequence of the no-cloning theorem; one cannot directly copy the quantum states being exchanged. Thus, to learn any information about the secret being exchanged over a \qkd{} link, one can only measure before the protocol is complete, and there is no ``transcript'' to attack passively after that point.

\paragraph{The HAKE framework.} A HAKE framework is put forward in the pioneering work due to Dowling, Brandt Hansen, and Paterson \cite{PQCRYPTO:DowHanPat20}, in which the $\Muckle$ protocol is proposed. $\Muckle$ is TLS-like, in the sense that unauthenticated handshakes are retrospectively authenticated when the transcript is authenticated. However, the authors argue that the use-cases they target are such that a PSK infrastructure is appropriate to achieve end-to-end authentication, and so, as we have discussed, their protocol is not an appropriate choice for a large-scale and dynamic network. As an enhancement, Bruckner, Ramacher and Striecks \cite{PQCRYPTO:BruRamStr23} proposed $\SigMuckle$, to efficiently derive an authenticated shared key even in large-scale quantum-safe networks. $\SigMuckle$ utilizes post-quantum signatures as the authentication mechanism.
As we have discussed, currently known post-quantum signature schemes are not as efficient as post-quantum $\KEM$s. Bruckner et al. already recognized such a property and left it as open problem to come up with more efficient authentication within HAKE.

In this work, we address this gap, and propose $\MuckleSharp$, a protocol that solely utilizes post-quantum $\KEM$s for authentication. In contrast with the signature case, only implicit authentication can be achieved directly with $\KEM$s, and further rounds of message-authentication code ($\MAC$) tags on the transcript of the protocol must be exchanged to yield explicit authentication. We provide a security proof of $\MuckleSharp$ within the HAKE framework\ifdefined\full.\else, and report a measurable improvement in efficiency compared to the performance of $\SigMuckle$.\fi

\paragraph{Related work.} The two works most closely related to this paper are \cite{PQCRYPTO:DowHanPat20} and \cite{PQCRYPTO:BruRamStr23}, in which the $\Muckle$ and $\SigMuckle$ protocols are put forward, respectively. Crucially, in \cite{PQCRYPTO:DowHanPat20}, the HAKE security model is defined, and it is in this model that the security of $\Muckle$, $\SigMuckle$ and $\MuckleSharp$ is proved. This model is a natural, hybrid-setting successor to the Bellare-Rogaway-style AKE models first put forward in \cite{bellare1993entity}.

Our protocol can be thought of as a synthesis of the $\SigMuckle$ protocol and the $\KEMTLS$ protocol defined in \cite{CCS:SchSteWig20}. This latter work, in turn, is an improvement on an earlier attempt to achieve signature-free AKE \cite{7467348}. We note also that, on the technical side, we make heavy use of the arguments in the security proof of the $\KEMTLS$ protocol, in particular relying on the standard identical-until-bad techniques of \cite{bellare2004code}.

There are various papers in the literature which seek to combine \pqc{} and \qkd{}. Mosca et al. \cite{mosca2013quantum} define a protocol in which the authentication channel required for \qkd{} is achieved with a post-quantum signature. They prove the security of this protocol in a different hybrid-setting AKE model, which is set-up to deal with quantum information (rather than treating the \qkd{} link as a black-box, as in our case). The protocol does not mix different secret sources, as in the $\Muckle$-related protocols. An experimental implementation of this kind of approach was achieved by Wang et al. in \cite{wang2021experimental}. Closer to the spirit of our own work, a recent paper proposes the ``Muckle++'' protocol \cite{garms2024experimental}, which offers similar large-network flexibility to $\SigMuckle$ by including a signature authentication mode. A security proof is not provided, but the authors achieve an experimental implementation of their protocol on commercial hardware.

\ifdefined\full
\paragraph{Contribution.} Our contribution can be summarized as follows. We carefully construct a HAKE protocol, dubbed $\KEMuckle$, in the framework of HAKE that uses KEM-based authentication instead of a signature-based one. Along the way, we propose adapted proof steps due to absence of digital signatures for authentication compared to $\SigMuckle$. Our new protocol has a slightly different message flow that we carefully align with the HAKE framework and which makes our changes to the $\SigMuckle$ non-trivial.
\else
\subsection{Contribution}
Our contribution can be summarized as follows:
\begin{itemize}
  \item We carefully construct a HAKE protocol, dubbed $\KEMuckle$, in the framework of HAKE that uses KEM-based authentication instead of a signature-based one. Along the way, we propose adapted proof steps due to absence of digital signatures for authentication compared to $\SigMuckle$. Our new protocol has a slightly different message flow that we carefully align with the HAKE framework and which makes our changes to the $\SigMuckle$ non-trivial.
  \item We implement $\KEMuckle$ in Python to evaluate its performance and bandwidth characteristics and to compare it against $\SigMuckle$. Our analysis shows that $\KEMuckle$ -- despite the additional communication cost -- benefits from the reduced computational complexity of the post-quantum secure KEMs compared to signature.
\end{itemize}
\fi

\subsection{More on the Technical Details}

In the following, we discuss the technical details related to the security proof, modes of authentication, and modeling the \qkd{} link in a more depth.

\paragraph{Security proof.} We prove the security of $\MuckleSharp$ within the HAKE model, making repeated use of the arguments put forward in the security proof of $\KEMTLS$. There are two main obstacles to overcome: fitting the $\KEMTLS$ arguments within a HAKE framework, and addressing the gap between implicit and explicit authentication that is introduced when one uses a $\KEM$ to authenticate. In \Cref{sec:kemuckle-sec}, we discuss how these technicalities are addressed.

\paragraph{Modes of authentication.} We briefly explain the difference in the authentication mechanisms of $\Muckle$, $\SigMuckle$ and $\MuckleSharp$. First, note that all three of those protocols begin with an unauthenticated exchange of ephemeral secrets; first of a classical $\KEM$, then of a post-quantum $\KEM$, and finally the exchange of a secret via a \qkd{} protocol. Eventually, one uses a series of pseudo-random function evaluations to derive a handshake secret. The authentication mechanisms then differ thus:
\begin{itemize}
    \item In the $\Muckle$ protocol it is assumed that any two parties in the network have already established a pre-shared key. They can then use this key to compute $\MAC$s on the transcript of the values exchanged, thereby authenticating the handshake secret.
    \item In the $\SigMuckle$ protocol, after the handshake secret is derived, the parties send (encrypted) certificates, which contain static, authenticated signing keys for a post-quantum signature scheme. Signatures of the transcript of the protocol are then computed with respect to these signing keys. After a final exchange of $\MAC$s (as is standard in such cases; see \cite{krawczyk2003sigma}), the handshake secret is authenticated.
    \item In our proposed $\KEMuckle$, after the handshake secret is defined, the same exchange of certificates takes place. This time, however, the certificates contain static, authenticated encapsulation keys for a post-quantum $\KEM$. Each entity then encapsulates an additional secret to their partner's long term encapsulation key, and upon decapsulation the resulting secret is, again via the application of pseudo-random function evaluations, ``baked in'' to the key schedule. Here, the authenticity is intuitively a result of the authenticity of the encapsulation key; since only the entity with the corresponding decapsulation key could recover the appropriate secret and adds this to the key scheduling. Crucially, only \textit{implicit} authentication is achieved here, and a final $\MAC$ exchange is required to make this authentication explicit.
\end{itemize}

\paragraph{Modelling the QKD link.}\label{sec:modelling-qkd}
As with the predecessors, $\Muckle$~\cite{PQCRYPTO:DowHanPat20} and $\SigMuckle$~\cite{PQCRYPTO:BruRamStr23}, we model the \qkd{} link out-of-band. More specifically, each entity has a black-box $\GetKey$ that gives them a key value $k_q$, and the mechanism by which this is achieved is abstracted to some $\qkd{}$ protocol running in the background. We model the potential failure of this \qkd{} link by giving the adversary access to a $\CorruptQK$ query, that reveals the key $k_q$.

Now, all current \qkd{} protocols require some form of auxiliary, symmetrically authenticated channel for their unconditional security to hold. Part of the motivation for moving away from the symmetric authentication of $\Muckle$, to public-key methods, was to remove the reliance on pre-shared key networks, which scale inherently badly. As such, there is some tension between the modeling assumptions we have made about the \qkd{} link, and the potential use cases of $\SigMuckle$ and $\MuckleSharp$ in more dynamic networks.

This gap is not an easy one to address. One of the reasons to model \qkd{} out-of-band is to avoid invoking the quantum-physical type arguments often seen in contemporary security proofs of \qkd{} protocols, and indeed, any method of addressing the gap will likely be required to feature such analysis. We consider this out of scope for our work, whose main purpose is to demonstrate efficiency gains with respect to $\SigMuckle$, while retaining provable security within the same framework. 

\paragraph{KEM-Based Authentication and Public-Key Infrastructures.}\label{par:pki} Available Public-Key Infrastructures (PKIs) will not be applicable in our setting due to certifying long-term signatures keys only whereas in \MuckleSharp we would need certifactes that include long-term encapsulation keys. The issue appears already in KEM-TLS which limits the use of such a protocol for Internet usage. 

However, since we are in a different setting and particularly motivated by the EuroQCI and related QKD-based networks, a novel PKI has to be developed alongside anyways to cope with the novel requirements in such quantum-safe networks (where such a new PKI could incorporate certifying long-term KEM keys as well in that case). Hence, we strongly believe that our \MuckleSharp protocol would definitely be applicable in such anticipated networks.

\section{Preliminaries}
In this section, we briefly recall notions related to (hybrid) authenticated key exchanges.

\paragraph{Notation.} 
    Let $\secpar \in \N$ be the security parameter. For a finite set $\setS$, we denote by $s\gets\setS$ the process of sampling an element $s$ uniformly from $\setS$. For an algorithm $A$, we let $y \gets A(\secpar,x)$ denote the process of running $A$ on input $(\secpar,x)$ with access to uniformly random coins, and assigning the result to $y$. (We may omit explicit mention of the $\secpar$-input and assume that all algorithms take $\secpar$ as input.) We say an algorithm $A$ is probabilistic polynomial time (PPT) if the running time of $A$ is polynomial in $\secpar$ by a probabilistic Turing machine. An algorithm $A$ is called quantum polynomial time (QPT) if it is a uniform family of quantum circuits with size polynomial in $\secpar$. A function $f$ is called negligible if its absolute value is smaller than the inverse of any polynomial, for sufficiently large input values (i.e., if $\forall c\in\N\ \exists k_0\ \forall \kk \geq k_0:|f(\kk)|<1/ \kk^c$). 

\subsection{Cryptographic Primitives and Schemes}
\begin{definition}[Pseudo-Random Function]
  Let $\prfeval\colon \mathcal{S} \times D \to \mathsf{R}$ be a family of functions and let $\Gamma$ be the set of all functions $D \rightarrow \mathsf{R}$. For a PPT distinguisher $\mathcal{D}$ we define the advantage function as
  \[
    \Adv^{\mathsf{PRF}}_{\mathcal{D},\prfeval}(\secpar) = \left| \Pr\left[s \getsr \mathcal{S}: {\cal D}^{\prfeval(s, \cdot)}(1^\secpar) = 1\right] - \Pr\left[f\getsr \Gamma: {\cal D}^{f(\cdot)}(1^\secpar) = 1\right] \right|
    \text{.}
  \]
  $\prfeval$ is a pseudorandom function (family) if it is efficiently computable and for all PPT distinguishers $\mathcal{D}$ there exists a negligible function $\varepsilon(\cdot)$ such that
  \[
    \Adv^{\mathsf{PRF}}_{\mathcal{D},\prfeval}(\secpar) \leq \varepsilon(\secpar)
    \text{.}
  \]
\end{definition}
A PRF $\prfeval$ is a dual PRF~\cite{EPRINT:BelLys15}, if $\mathcal{G}: D \times \mathcal{S} \to \mathsf{R}$ defined as $\mathcal{G}(d, s) = \prfeval(s, d)$ is also a PRF.

We recall the notion of message authentication codes (MACs) as well as digital signature schemes, and the standard unforgeability notions below.
\begin{definition}[Message Authentication Codes] \label{def:macs}
  A message authentication code $\MAC$ is a triple $(\kgen,\allowbreak \ssign, \sverify)$ of PPT algorithms, which are defined as:
  \begin{description}
    \item[$\kgen(1^\secpar)\colon$] This algorithm takes a security parameter $\secpar$ as input and outputs a secret key $\sk$.
    \item[$\msign(\sk, \msg)\colon$] This algorithm takes a secret key $\sk \in \Key$ and a $\msg \in \Msg$, and outputs an authentication tag $\tau$.
    \item[$\mverify(\sk,\msg,\tau)\colon$] This algorithm takes a secret key $\sk$, a message $\msg \in \Msg$, and an authentication tag $\tau$ as input, and outputs a bit $b \in \{0,1\}$.
  \end{description}
\end{definition}
A MAC is correct if for all $\secpar \in \N$, for all $\sk \gets \kgen(1^\secpar)$ and for all $\msg \in \Msg$, it holds that
\[
  \Pr\left[ \mverify(\sk, \msg, \allowbreak\msign(\sk,\msg))=1 \right] =1
  \text{,}
\] where the probability is taken over the random coins of $\kgen$ and $\msign$.
\begin{definition}[$\EUFCMA$ security of \MAC] For a PPT adversary $\advA$, we define the advantage function in the sense of existential unforgeability under chosen message attacks ($\EUFCMA$) as
  \[
    \AdvEUFCMA{\advA,\MAC}(1^\secpar) = \Pr\left[ \ExpEUFCMA{\advA,\MAC}(1^\secpar) = 1\right]
    \text{,}
  \]
  where the corresponding experiment is depicted in \Cref{fig:mac-unfcma}. If for all PPT adversaries $\adv$ there is a negligible function $\varepsilon(\cdot)$ such that
  \(
    \AdvEUFCMA{\advA,\MAC}(1^\secpar) \leq \varepsilon(\secpar)
    \text{,}
  \)
  we say that $\MAC$ is $\EUFCMA$ secure.
\end{definition}
\begin{experiment}[ht]
\centering
\pseudocode[mode=text]{
$\ExpEUFCMA{\advA,\MAC}(1^\secpar)$: \\
\t $\sk \gets \kgen(1^\secpar)$, $\mathcal{Q} \gets \emptyset$\\
\t $(\msg^*, \tau^*)\gets\advA^{\msign',\mverify'}(1^\secpar)$\\
\t \t where oracle $\msign'(\msg)$:\\
\t \t \t $\mathcal{Q} \gets \mathcal{Q} \cup \{\msg \}$\\
\t \t \t return $\msign(\sk, \msg)$\\
\t \t where oracle $\mverify'(\msg, \tau)$:\\
\t \t \t return $\mverify(\sk, \msg, \tau)$\\
\t return $1$, if $\mverify(\sk, \msg^*, \tau^*) = 1~\land~\msg^*\notin \mathcal{Q}$, return $0$, otherwise \\
}
\caption{$\EUFCMA$ security experiment for a MAC $\MAC$.}
\label{fig:mac-unfcma}
\end{experiment}

We recall the notion of key-encapsulations mechanisms (KEMs), and the standard chosen-plaintext and chosen-ciphertext notions below.
\begin{definition}[Key-Encapsulation Mechanism]
A key-encapsulation mechanism scheme $\KEM$ with key space $\Key$ consists of the three PPT algorithms \((\kgen,\Encaps,\Decaps)\):
\begin{description}
  \item[$\kgen(1^\secpar)\colon$] This algorithm takes a security parameter $\secpar$ as input, and outputs public and secret keys \((\pk,\sk)\).
  \item[$\Encaps(\pk)\colon$] This algorithm takes a public key $\pk$ as input, and outputs a ciphertext \(\ctxt\) and key $\key$.
  \item[$\Decaps(\sk, \ctxt)\colon$] This algorithm takes a secret key $\sk$ and a ciphertext $\ctxt$ as input, and outputs \(\key\) or \(\{\bot\}\).
\end{description}
\end{definition}
We call a $\KEM$ correct if for all \(\secpar\in\N\), for all \((\pk,\sk)\gets\kgen(\secpar)\), for all \((\ctxt,\key)\gets\Encaps(\pk)\), we have that
\[
  \Pr[\Decaps(\sk,\ctxt)=\key] = 1 \text{,}
\] where the probability is taken over the random coins of $\kgen$ and $\Encaps$.
\begin{definition}[IND-CPA and IND-CCA security of \KEM]
  For a PPT adversary $\advA$, we define the advantage function in the sense of indistinguishability under chosen-plaintext attacks (\KEMINDCPA) and indistinguishability under chosen-ciphertexts attacks (\KEMINDCCA) as
  \begin{align*}
    \AdvKEMINDCPA{\advA,\KEM}(1^\secpar) &= \left| \Pr\left[ \ExpKEMINDCPA{\advA,\KEM}(1^\secpar) = 1\right] -\frac{1}{2} \right|
    \text{, and} \\
    \AdvKEMINDCCA{\advA,\KEM}(1^\secpar) &= \left| \Pr\left[ \ExpKEMINDCCA{\advA,\KEM}(1^\secpar) = 1\right] -\frac{1}{2} \right|
  \end{align*}
  where the corresponding experiments are depicted in \Cref{fig:kem-ind-cca}. If for all PPT adversaries $\adv$ there is a negligible function $\varepsilon(\cdot)$ such that
  \[
    \AdvKEMINDCPA{\advA,\KEM}(1^\secpar) \leq \varepsilon(\secpar)
    \text{ or }
    \AdvKEMINDCCA{\advA,\KEM}(1^\secpar) \leq \varepsilon(\secpar),
  \]
  then we say that $\KEM$ is \KEMINDCPA or \KEMINDCCA secure, respectively.
\end{definition}
\begin{experiment}[t]
\centering
\pseudocode[mode=text]{
$\ExpKEMINDT{\advA,\KEM}(\secpar)$: \\
\t $(\pk, \sk) \gets \kgen(1^\secpar)$\\
\t \((\ctxt^*,\key_0)\gets\Encaps(\pk),\key_1\getsr\Key\)\\
\t $\mathcal{Q} \gets \emptyset$, \(b\getsr\{0,1\}^\secpar\)\\
\t $b^* \gets\advA^{\mathcal{O}}(\pk, \ctxt^*, \key_b)$\\
\t \t where $\mathcal{O} = \{\Decaps'\}$ if $T = \mathsf{cca}$ with oracle $\Decaps'(\ctxt)$:\\
\t \t \t $\mathcal{Q} \gets \mathcal{Q} \cup \{\ctxt \}$\\
\t \t \t return $\Decaps(\sk, \ctxt)$\\
\t return $1$, if $b = b^*~\land~\ctxt^*\notin \mathcal{Q}$, return otherwise $0$ \\
}
\caption{IND-CPA and IND-CCA security experiments for \KEM with $T \in \{\mathsf{cpa}, \mathsf{cca}\}$.}
\label{fig:kem-ind-cca}
\end{experiment}

We further use the Authenticated Encryption with Associated Data (AEAD) scheme as defined in \cite{CCS:RBBK01}.

\begin{remark}
    We can straightforwardly define post-quantum security for our definitions by requiring that the definitions hold against QPT adversaries.
\end{remark}

\subsection{Hybrid Authenticated Key Exchange}

We recall the hybrid authenticated key exchange (HAKE) security model~\cite{PQCRYPTO:DowHanPat20,PQCRYPTO:BruRamStr23}. The HAKE security experiment $\mathsf{Exp}^{\mathsf{hake},\clean}_{\advA,\Pi,n_P,n_S,n_T}$ is described as in \cite[Fig. 5, App. C]{PQCRYPTO:DowHanPat20}. Here, we only recall the execution environment, adversarial interaction, and matching sessions.

\paragraph{Execution environment.}
We consider a set of $n_P$ parties $P_1, \ldots, P_{n_P}$ which are able to run up to $n_S$ sessions of a key-exchange protocol between them, where each session may consist of $n_T$ different stages of the protocol\footnote{Notice the terminology here is different to that in, say, the analysis of TLS 1.3~\cite{JC:DFGS21}. What we call a \textit{stage} would in that context be called a \textit{session}, and we do not have any subdivision below the session level.}. A ``stage'' implicitly describes a pair of stages, one local to the initiator and one local to the responder. A stage is said to have ``accepted'' if it has completed its key derivation schedule without aborting the protocol. Each party $P_i$ has access to its long-term key pair $(\pk_i, \sk_i)$ and to the public keys of all other parties. Each session is described by a set of session parameters:
\begin{itemize}
  \item $\rho \in \{\mathsf{init},\mathsf{resp}\}$: The role (initiator or responder) of the party during the current session.
  \item $pid \in n_P$: The communication partner of the current session.
  \item $stid \in n_T$: The current stage of the session.
  \item $\alpha \in \{\act,\acc,\rej,\bot\}$: The status of the session. Initialized with $\bot$.
  \item $m_i[stid], i \in \{s,r\}$: All messages sent ($i=s$) or received ($i=r$) by a session up to the stage $stid$. Initialized with $\bot$.
  \item $k[stid]$: All session keys created up to stage $stid$. Initialized with $\bot$.
  \item $exk[stid], x \in \{q,c,s\}$: All ephemeral post-quantum (q), classical (c) or symmetric (s) secret keys created up to stage $stid$. Initialized with $\bot$.
  \item $pss[stid]$: The per-session secret state (\SecState) that is created during the stage $stid$ for the use in the next stage.
  \item $st[stid]$: Storage for other states used by the session in each stage.
\end{itemize}

We describe the protocol as a set of algorithms $(f, \kgen{}XY, \kgen{}ZS)$:
\begin{itemize}
  \item $f(\secpar, \pk_i, \sk_i, pskid_i, psk_i,\pi,m) \rightarrow (m',\pi')$: A probabilistic algorithm that represents an honest execution of the protocol. It takes a security parameter $\secpar$, the long-term keys $(\pk_i, \sk_i)$, the session parameters $\pi$ representing the current state of the session, and a message $m$, and outputs a response $m'$ and the updated session state $\pi'$.
  \item $\kgen{}XY(\secpar) \rightarrow (\pk,\sk)$: A probabilistic asymmetric key-generation algorithm that takes a security parameter $\secpar$ and creates a public-secret-key pair $(\pk,\sk)$. $X \in \{E,L\}$ determines whether the created key is an ephemeral (E) or a long-term (L) secret key. $Y \in \{Q,C\}$ determines whether the key is classical (C) or post-quantum (Q).
  \item $\kgen{}ZS(\secpar) \rightarrow (psk,pskid)$: A probabilistic symmetric key-generation algorithm that takes a security parameter $\secpar$ and outputs symmetric keying material $(psk)$. $Z \in \{E,L\}$ determines whether the created key is an ephemeral (E) or a long-term (L) secret key.
\end{itemize}

For each party $P_1, \ldots, P_{n_P}$, classical as well as post-quantum long-term keys are created using the corresponding $\kgen{}XY$ algorithms. The challenger then queries a uniformly random bit $b \gets \{0,1\}$ that will determine the key returned by the $\Test$ query. From this point on, the adversary may interact with the challenger using the queries defined in the next section. At some point during the execution of the protocol, the adversary $\advA$ may issue a single $\Test$ query\footnote{Unlike similar works in the area we explicitly limit our adversary to making exactly one $\Test$ query. The reason is for the sake of simplicity: proofs of protocols in which the adversary can issue many $\Test$ queries begin by guessing if a particular session was tested, and aborting otherwise, thereby allowing the proof to argue about a single tested session. By simply limiting to a single test query, the advantage derived in our proof will differ only by a constant factor from the advantaged derived for an equivalent proof in which multiple test queries were permitted. Since proofs of this kind are highly non-tight anyway, we consider this an acceptable sacrifice in the service of readability.} and present a guess for the value of $b$. If $\advA$ guesses correctly and the session satisfies the cleanness predicate, the adversary wins the key-indistinguishability experiment.

\paragraph{Adversarial Interaction.}
The HAKE framework defines a range of queries that allow the attacker to interact with the communication:
\begin{itemize}
  \item $\Create(i,j,role) \rightarrow \{(s),\bot\}$: Initializes a new session between party $P_i$ with role $role$ and the partner $P_j$. If the session already exists, then the query returns $\bot$, otherwise the session $(s)$ is returned.
  \item $\Send(i,s,m) \rightarrow \{m',\bot\}$: Enables $\advA$ to send messages to sessions and receive the response $m'$ by running $f$ for the session $\pi_i^s$. Returns $\bot$ if the session is not active.
  \item $\Reveal(i,s,t)$: Provides $\advA$ with the session keys, corresponding to a session $\pi_i^s$ in stage $t$, if the session is in the accepted state. Otherwise, $\bot$ is returned.
  \item $\Test(i,s,t) \rightarrow \{k_b,\bot\}$: Provides $\advA$ with the real (if $b=1$) or random ($b=0$) session key, for the session $\pi_i^s$ in stage $t$, for the key-indistinguishably experiment.
  \item $\Corrupt{}XY(i) \rightarrow \{key,\bot\}$: Provides $\advA$ with the long-term $XY \in \{\mathsf{SK},\mathsf{QK},\mathsf{CK}\}$ keys for $P_i$. If the key has been corrupted previously, then $\bot$ is returned. Specifically:
    \begin{itemize}
    \item $\CorruptSK$: Reveals the long-term symmetric secret (if available).
    \item $\CorruptQK$: Reveals the post-quantum long-term key (if available).
    \item $\CorruptCK$: Reveals the classical long-term key (if available).
    \end{itemize}
  \item $\Compromise{}XY(i,s,t) \rightarrow \{key,\bot\}$: Provides $\advA$ with the ephemeral $XY \in \{ \mathsf{QK}, \mathsf{CK}, \mathsf{SK}, \mathsf{SS} \}$ keys created during the session $\pi_i^s$ prior to stage $t$. If the ephemeral key has already been compromised, then $\bot$ is returned. Specifically:
    \begin{itemize}
    \item $\CompromiseQK$: Reveals the ephemeral post-quantum key.
    \item $\CompromiseCK$: Reveals the ephemeral classical key.
    \item $\CompromiseSK$: Reveals the ephemeral quantum key.
    \item $\CompromiseSS$: Reveals the ephemeral per session state (\SecState).
    \end{itemize}
\end{itemize}

\paragraph{Matching sessions.}
Furthermore, we recall the definitions of matching sessions~\cite{PROVSEC:LKZC07} and origin sessions~\cite{ESORICS:CreFel12} which covers that the two parties involved in a session have the same view of their conversation.
\begin{definition}[Matching sessions]
We consider two sessions $\pi_i^s$ and $\pi_j^r$ in stage t to be matching if all messages sent by the former session $\pi_i^s.m_s[t]$ match those received by the later $\pi_j^r.m_r[t]$ and all messages sent by the later session $\pi_j^r.m_s[t]$ are received by the former $\pi_i^s.m_r[t]$.

$\pi_i^s$ is considered to be prefix-matching with $\pi_j^r$ if $\pi_i^s.m_s[t] = \pi_j^r.m_r[t]'$ where $\pi_j^r.m_r[t]$ is truncated to the length of $\pi_i^s.m_s[t]$ resulting in $\pi_j^r.m_r[t]'$.
\end{definition}
\begin{definition}[Origin sessions]
We consider a session $\pi_i^s$ to have an origin session with $\pi_j^r$ if $\pi_i^s$ matches $\pi_j^r$ or if $\pi_i^s$ prefix-matches $\pi_j^r$.
\end{definition}

\paragraph{Implicit vs. explicit authentication.} KEM authentication by itself achieves only ``implicit authentication''. Intuitively, the idea is that as long as KEM security holds and long-term secrets are not prematurely compromised, only the intended party can recover the intended secret. The authentication is implicit because, under these assumptions, any keying material derived from this implicitly authenticated secret must only be available to the intended party. We do not, however, get a guarantee that the intended party has indeed accessed this secret -- so-called ``explicit authentication''. Unlike in the signature case, which provides authentication and proof that the intended party is ``live'' in the same step, this gap must be addressed. This is the reason for the $\MAC$ exchange at the end of the protocol -- it provides a proof that the partner is `live'. The strategy of the proof is to replace the secret that was encapsulated against the static public key, by a random secret, bounding the difference in advantage between these two games by the advantage in a $\INDCCA1$ game. We can then argue that a peer who accepted the final $\MAC$ tag without an honest partner session must have accepted a forgery, which allows us to bound the final advantage by an advantage in a $\MAC$ unforgeability game. Note that the case in which an honest session accepts the final $\MAC$ tag without an honest a partner is called ``malicious acceptance'' in \cite{schwabe2020post}.

\paragraph{HAKE security.}
Dowling et al.~\cite{PQCRYPTO:DowHanPat20} define key indistinguishability (i.e., what we dub HAKE security) with respect to a predicate $\clean$. However, their predicate is specific to Muckle and, hence, we therefore only give the formal security definition next and postpone the discussion of the predicate to \Cref{sec:kemuckle-sec}.

\begin{definition}[HAKE security]
  Let $\Pi$ be a key-exchange protocol and $n_P, n_S, n_T \in \N$. For a predicate $\clean$ and an adversary $\advA$, we define the advantage of $\advA$ in the HAKE key-indistinguishability game as
  \[
    \Adv^{\mathsf{hake},\clean}_{\advA,\Pi,n_P,n_S,n_T}(\secpar) = \left| \Pr\left[ \mathsf{Exp}^{\mathsf{hake},\clean}_{\advA,\Pi,n_P,n_S,n_T}(\secpar) = 1 \right] \right|\text{.}
  \]
  We say that $\Pi$ is HAKE-secure if $\Adv^{\mathsf{hake},\clean}_{\advA,\Pi,n_P,n_S,n_T}(\secpar)$ is negligible in the security parameter $\secpar$ for all $\advA$. 
  
  If security also holds against any QPT adversary $\advA$, then we call $\Pi$ post-quantum secure.
\end{definition}

\section{Muckle\#: A HAKE Protocol With KEM-Based Authentication}\label{sec:kemuckle}

In this section, we present our novel variant $\KEMuckle$. To recall, the $\SigMuckle$ protocol \cite{PQCRYPTO:BruRamStr23} combines conventional, \pqc{}, and \qkd{} keys through the use of a key derivation function. More concretely, $\SigMuckle$ requires conventional and post-quantum KEMs, a \qkd{} mechanism as well as a digital signature scheme to create the final shared secret. $\KEMuckle$ on the other side is similar to $\SigMuckle$, but utilizes post-quantum KEMs instead of signatures in the authentication step. \Cref{fig:sig_muckle_stage} depicts one stage of the $\KEMuckle$ protocol that overall runs in several stages.

\begin{figure}
\centering
\scriptsize
\begin{tikzpicture}\node{
  \begin{tabular}{p{4.5cm} >{\centering}p{2cm} p{4.5cm}}
  \centering
  Initiator & & \centering Responder \tabularnewline
  \hline
  \centering $\sk_I$, $cert_I$, $pid_I$ \SecState  & & \centering $\sk_R$, $cert_R$, $pid_R$ \SecState \tabularnewline
   & & \tabularnewline
  $n_I \getsr \{0,1\}^{\secpar}$ & & \tabularnewline
  $\pk_c, \sk_c \gets \KEM_c.\kgen(1^\secpar)$ & & \tabularnewline
  $\pk_{pq}, \sk_{pq} \gets \KEM_{pq}.\kgen(1^\secpar)$ & & \tabularnewline
   & $\xrightarrow{\makebox[2cm]{$\bm{m_1}\colon \pk_c, \pk_{pq}, n_I$}}$ & \raggedleft $n_R \getsr \{0,1\}^{\secpar}$ \tabularnewline
   & & \raggedleft $\ctxt_c, ss_c \gets \KEM_c.\Encaps(\pk_c)$ \tabularnewline
   & & \raggedleft $\ctxt_{pq}, ss_{pq} \gets \KEM_{pq}.\Encaps(\pk_{pq})$ \tabularnewline
   & $\xleftarrow{\makebox[2cm]{$\bm{m_2}\colon \ctxt_c, \ctxt_{pq}, n_R$}}$ & \tabularnewline
   $ss_c \gets \KEM_c.\Decaps(\sk_c, \ctxt_c)$ & & \tabularnewline
   $ss_{pq} \gets \KEM_{pq}.\Decaps(\sk_{pq}, \ctxt_{pq})$ & & \tabularnewline
   & & \tabularnewline
   \multicolumn{3}{c}{$k_c \gets \prfeval(ss_c,\lbl_0\|H_1)$}  \tabularnewline
   \multicolumn{3}{c}{$k_{pq} \gets \prfeval(ss_{pq},\lbl_1\|H_1)$} \tabularnewline
   & & \tabularnewline
   \multicolumn{3}{c}{$k_0\gets\prfeval(k_{pq},\lbl_2\|H_1)$}  \tabularnewline
   \multicolumn{3}{c}{$k_1\gets\prfeval(k_c,\lbl_3\|k_0)$} 
   \tabularnewline
   & & \tabularnewline
$k_q\gets\GetKey_{qkd}(1^\secpar, pid_I)$ & &\hfill $k_q\gets\GetKey_{qkd}(1^\secpar, pid_R)$\tabularnewline
   \multicolumn{3}{c}{$k_2\gets\prfeval(k_q,\lbl_4\|k_1)$} \tabularnewline
   \multicolumn{3}{c}{$k_3\gets\prfeval(\SecState,\lbl_5\|k_2)$} \tabularnewline
    & & \tabularnewline
   \multicolumn{3}{c}{$\IHTS \gets \prfeval(k_3,\lbl_7\|H_1)$} \tabularnewline
   \multicolumn{3}{c}{$\RHTS \gets \prfeval(k_3,\lbl_8\|H_1)$} \tabularnewline
    & & \tabularnewline
   \multicolumn{3}{c}{$\dHS \gets \prfeval(k_3,\lbl_6\|H_0)$} \tabularnewline
   & \tabularnewline

   Verify $cert_R[\pk_R]$ & $\xleftarrow{\makebox[2cm]{$\bm{m_3}\colon \{cert_R[\pk_R]\}_{\RHTS}$}}$ \tabularnewline
   & & \tabularnewline
   $(\ctxt_I,ss_I)\gets\KEM_s.\Encaps(\pk_R)$ \tabularnewline
   & \tabularnewline
   & $\xrightarrow{\makebox[2cm]{$\bm{m_4}\colon \{\ctxt_I\}_{\IHTS}$}}$ \tabularnewline
   & & \raggedleft $ss_I\gets\KEM_s.\Decaps(\sk_R,\ctxt_I)$ \tabularnewline
   & & \tabularnewline
   \multicolumn{3}{c}{$\AHS \gets \prfeval(\dHS,\lbl_9\|ss_I)$}  \tabularnewline
   & & \tabularnewline
   \multicolumn{3}{c}{$\IAHTS \gets \prfeval(\AHS,\lbl_{10}\|H_2)$} \tabularnewline
   \multicolumn{3}{c}{$\RAHTS \gets \prfeval(\AHS,\lbl_{11}\|H_2)$} \tabularnewline
   & & \tabularnewline
   \multicolumn{3}{c}{$\dAHS \gets \prfeval(\AHS,\lbl_{12}\|H_0)$} \tabularnewline
   & & \tabularnewline
   & $\xrightarrow{\makebox[2cm]{$\bm{m_5}\colon \{cert_I[\pk_I]\}_{\IAHTS}$}}$
   & \raggedleft Verify $cert_I[\pk_I]$ \tabularnewline
   \tabularnewline
   & & \raggedleft $(\ctxt_R,ss_R)\gets\KEM_s.\Encaps(\pk_R)$ \tabularnewline
   & $\xleftarrow{\makebox[2cm]{$\bm{m_6}\colon \{\ctxt_R\}_{\RAHTS}$}}$ \tabularnewline
   $ss_R\gets\KEM.\Decaps(\sk_I,\ctxt_R)$ \tabularnewline
   & & \tabularnewline
   \multicolumn{3}{c}{$\MS\gets\prfeval(\dAHS,\lbl_{13}\|ss_R)$} \tabularnewline
   & & \tabularnewline
   \multicolumn{3}{c}{$\fk_I\gets\prfeval(\MS,\lbl_{14}\|H_3)$} \tabularnewline
   \multicolumn{3}{c}{$\fk_R\gets\prfeval(\MS,\lbl_{15}\|H_3)$} \tabularnewline
   $\IF \gets\MAC.\msign(\fk_I,H_3)$ \tabularnewline
   & $\xrightarrow{\makebox[2cm]{$\bm{m_7}\colon \{\IF\}_{\IAHTS}$}}$ & \raggedleft Abort if $\MAC.\mverify(\fk_I,H_3,\IF)\stackrel{?}{=}0$\tabularnewline
   & & \tabularnewline
  \multicolumn{3}{c}{$\IATS\gets\prfeval(\MS,\lbl_{16}\|H_4)$} \tabularnewline
   & & \raggedleft $\RF \gets \MAC.\msign(\fk_R,H_4)$ \tabularnewline
   Abort if $\MAC.\mverify(\fk_R,H_4,\RF)\stackrel{?}{=}0$ & $\xleftarrow{\makebox[2cm]{$\bm{m_8}\colon \{\RF\}_{\RAHTS}$}}$ \tabularnewline
   & & \tabularnewline
   \multicolumn{3}{c}{$\RATS\gets\prfeval(\dAHS,\lbl_{17}\|H_5)$} \tabularnewline
   \multicolumn{3}{c}{$\SecState\gets\prfeval(\MS,\lbl_{18}\|H_5)$} \tabularnewline
  \end{tabular}
};
\end{tikzpicture}
\caption{One stage of the \KEMuckle protocol. We have a classical KEM $\KEM_c$, a post-quantum KEM $\KEM_{pq}$, a MAC $\MAC$, a pseudorandom function $\prfeval$, a post-quantum KEM $\KEM_s$ used for authentication, and a symmetric key $k_q$ obtained from the \qkd{} component (provided out-of-band via function $\GetKey_{qkd}$). We assume that the certificates $cert_I$ and $cert_R$ contain the long-term public keys $\pk_I$ and $\pk_R$, respectively. Messages $\bm{m_i}\colon \{ \msg_{i,1}, \ldots \}_k$ denote that $\msg_{i,1}, \ldots$ is encrypted with an authenticated encryption scheme using the secret key $k$. The various contexts and labels are given in \Cref{tab:contexts,tab:labels}.} 
\label{fig:sig_muckle_stage}
\end{figure}

\begin{table}[t]
\centering
\caption{Values for the contexts used in the \KEMuckle key schedule. The context inputs follow the choices in the TLS 1.3 handshake~\cite{JC:DFGS21}.}
\label{tab:contexts}
\begin{tabular}{c|l|c|l}
  \toprule
Label & Context Input & Label & Context Input \\
\midrule
$H_\varepsilon$ & ``'' &
$H_0$ & $H(\text{``''})$ \\
$H_1$ & $H(\bm{m_1} \| \bm{m_2})$ &
$H_2$ & $H(\bm{m_1} \| \ldots \| \bm{m_4})$ \\
$H_3$ & $H(\bm{m_1} \| \ldots \| \bm{m_6})$  &
$H_4$ & $H(\bm{m_1} \| \ldots \| \bm{m_7})$  \\
$H_5$ & $H(\bm{m_1} \| \ldots \| \bm{m_8})$  &
 \\
\bottomrule
\end{tabular}
\end{table}

\begin{table}[t]
\centering
\caption{Values for the labels used in the \SigMuckle key schedule for domain separation. Some of these labels are directly based on the corresponding labels in the TLS 1.3 handshake~\cite{JC:DFGS21}. The concrete value of these labels is unimportant as long as they are unique.}
\label{tab:labels}
\begin{tabular}{c|l|c|l}
  \toprule
Label & Label Input & Label & Label Input \\
\midrule
$\lbl_0$ & ``derive k c'' &
$\lbl_1$ & ``derive k pq'' \\
$\lbl_2$ & ``first ck'' &
$\lbl_3$ & ``second ck'' \\
$\lbl_4$ & ``third ck'' &
$\lbl_5$ & ``fourth ck'' \\
$\lbl_6$ & ``i hs traffic'' &
$\lbl_7$ & ``r hs traffic'' \\
$\lbl_8$ & ``hs derived'' &
$\lbl_9$ & ``first ak''\\
$\lbl_{10}$ & ``i ahs traffic'' & \lbl{11} & ``r ahs traffic'' \\
$\lbl_{12}$ & ``ahs derived'' & $\lbl_{13}$ & ``second ak'' \\
$\lbl_{14}$ & ``derive i fk'' & $\lbl_{15}$ & ``derive r fk'' \\
$\lbl_{16}$ & ``i app traffic'' & $\lbl_{17}$ & ``r app traffic'' \\
$\lbl_{18}$ & ``secstate'' & \\
 
\bottomrule
\end{tabular}
\end{table}

The protocol flow is as follows. The initiator and the responder in the initial stage possess secret decryption keys $\sk_I$ and $\sk_R$ of a post-quantum KEM, respectively, and a secret state $\SecState$ (which is initially set to an empty string). 

The initiator and the responder further possess certificates $cert_I$ and $cert_R$, respectively, from a Public-Key Infrastructure (which is out of scope of this work). The certificates contain the verification keys $\pk_I$ and $\pk_R$ (available to all entities), respectively. 

The initiator and the responder have access to the labels defined in \Cref{tab:labels}. Those labels essentially assure separation of the key derivation function (\KDF) input (which is depicted as $\prfeval$). The initiator chooses a random bitstring $n_I$ with bit-length greater or equal than 128 bits. The initiator generates ephemeral public-secret key pairs $(\pk_{c}, \sk_{c})$, $(\pk_{pq}, \sk_{pq})$ via the key generation of the classic and post-quantum KEM, respectively. (If no public-secret key pair $(\pk_c, \sk_c)$ is generated, then the initiator sets $pk_c$ and $sk_c$ to empty strings.)

The initiator sends the public keys and the random values to the responder via a public channel.

The responder chooses a random bitstring $n_R$ with bit-length $k$ (with $k$ greater or equal than 128 bits). The responder generates an encapsulation $(c_{pq}, ss_{pq})$ via the encapsulation algorithm of the post-quantum KEM using the public keys $pk_{pq}$. If $pk_c$ is available at the responder, the responder generates an encapsulation $(c_c, ss_c)$ via the encapsulation algorithm of the classical KEM using the public key $pk_c$. If $pk_c$ is an empty string, then $c_c$ and $ss_c$ are set as empty strings.

The responder sends $c_{pq}$, $c_c$, and $n_R$ to the initiator via a public channel. The initiator computes the decapsulation $ss_{pq}$ and potentially $ss_c$ using the decapsulation algorithm of the post-quantum KEM and the classical KEM using the secret key $\sk_{pq}$ and $\sk_c$ with the ciphertexts $c_{pq}$ and $c_c$, respectively. If $c_c$ is an empty string, then $ss_c$ is set as an empty string. 

Now, the initiator and responder invoke a series of \KDF calls as given in \Cref{fig:sig_muckle_stage} where $k_q$ is the symmetric \qkd{} key. This essentially ensures a cryptographically sound way of ``binding'' the different keys and \SecState together (which will argue about in the security proof).

The responder computes the encryption of $cert_R$ using the encryption algorithm of the AEAD with encryption key $\RHTS$ (and with associated data string ``Message 3'') and must send the resulting ciphertext $m_3$ to the initiator. The initiator computes the decryption of the received ciphertext $m_3$ using the decryption algorithm of the AEAD with decryption key $\RHTS$ and must verify the certificate $cert_R$ (where latter is out of scope of this work). Both parties have a value \dHS.

Now, the authentication phase starts. The initiator computes the encapsulation ciphertext and key $c_I,ss_I$ under the retrieved long-term public key $\pk_R$ of the receiver using the encapsulation algorithm of the post-quantum KEM and sends the $c_I$ (AEAD-encrypted under \IHTS) to the receiver.

The receiver computes the decapsulation key $ss_I$ under the retrieved ciphertext $c_I$ (AEAD-decrypted under \IHTS) using the decapsulation algorithm of the post-quantum KEM. Both parties are able to retrieve the value $\dAHS$ via several calls to the \KDF (see that also $ss_I$ is given as input to one \KDF call). The correctness of the post-quantum KEM guarantees that starting from the same value \dHS, the initiator and receiver come to the same value \dAHS \textit{and} the responder is authenticated.

Next, the protocol does the same analogously for the initiator and, hence, both derive at the ``master value'' value \MS \textit{and} are mutually authenticated. What is missing is the key confirmation which is carried out on both sides via \MAC key derived from \MS. Finally, after successful mutual key confirmations, the secure state \SecState gets updated via the \KDF using the master value \MS. The secure state will input to the next stage of the protocol (once triggered).

Under the assumption of the correctness of the KEM and MAC building blocks, correctness of \MuckleSharp follows.

\paragraph{Binding, Intended Partners, and the QKD Link.} Each party in a session of \MuckleSharp has an input parameter indicating their intended session partner. Implicitly, when certificates are received, this is the value that the parties are checking against.

On the other hand, as discussed, we model the entire QKD link out-of-band, including its authentication. As such, in our call to the QKD link we need to explicitly specify the intended partner with whom we wish to establish a QKD key.

Notice that in the pre-shared key variant of TLS 1.3~\cite{PQCRYPTO:DowHanPat20}, parties are required to ensure, at the start of the protocol, that they are working from the same pre-shared key. This is accomplished by deriving a ``binding key'', that is used to compute a $\MAC$ of some pre-arranged label. We do not require this property if we are happy to insist that our long-term authentication secrets are not compromised before the session completes, since the authentication mechanism will detect a disagreement between the two parties on the key established by the QKD link. However, if we wish to allow everything but the QKD link to fail, we would require the derivation of a binding key to ensure that each local copy of the session is using secrets derived from the same key established by the QKD link. We refer the reader to the pre-shared key variant of TLS 1.3 for brief details how this would be done.

\subsection{Security of \KEMuckle}\label{sec:kemuckle-sec}

The security proof is inspired by ideas in the realm of quantum-safe key exchanges. In particular we will use a security game in the tradition of the Bellare-Rogaway style key-indistinguishability games \cite{bellare1993entity}, adapted to the hybrid setting in \cite{PQCRYPTO:DowHanPat20}. Roughly speaking, we consider an adversary who controls \textit{all} communication between honest parties, and has access to various oracles by which it can learn stage keys, ephemeral secrets, long-term secrets, and so on; the idea is that the adversary should not be able to do any better than simply acting passively as a wire, and transmitting the real information sent by the parties.

The HAKE framework was presciently set up to cover a wide variety of protocols in the hybrid setting. As such we do not have to adapt the framework to cater for our protocol -- for example, the compromise of long-term KEM keys is the same as the compromise of long-term signature keys from the point of view of the relevant HAKE queries. Our security proof, however, also draws heavily from the techniques used in the security proof of \KEMTLS~\cite{CCS:SchSteWig20}. Unlike our analysis, the security proof of \KEMTLS treats that protocol as establishing several internal ``stage keys,'' each of which can be revealed via queries available to the adversary. By contrast, we only deal with the indistinguishability of the final master secret; nevertheless, in the course of our security proof, we will show that the keys used internally for symmetric encryption, which could reasonably be described as ``stage keys'' in the model of \KEMTLS, are indistinguishable from random.

\paragraph{The $\KEMuckle$ cleanness predicate.} As is standard in this area, we need to define a permissible set of queries which the adversary is allowed to issue. We do this by considering 
a stage ``fresh'' only if certain queries or combinations of queries have not been issued, and aborting if the tested stage is not fresh. This is both to cover trivial breaches -- for example, a $\Reveal$ and a $\Test$ query being issued on the same stage -- but also to define the scope of robustness of the protocol; that is, which secrecy breaches we can tolerate while maintaining a small advantage in the HAKE key-indistinguishability experiment. 

The list of combinations of queries that prevent the freshness of a stage is called the ``cleanness predicate;'' both of \cite{PQCRYPTO:DowHanPat20} and \cite{PQCRYPTO:BruRamStr23} define bespoke cleanness predicates to cater to their particular authentication mechanisms. We define a cleanness predicate tailored for our context. In this direction, we note that an adversary $\advA$ has access to all queries defined in the HAKE framework. As no pre-shared key exists in the (first-time version of the) $\KEMuckle$ protocol, the query $\CorruptSK$ will return $\perp$ if called.

The predicate itself is, aside from the trivial controls on $\Reveal$ queries, divided into a passive and active case (as discussed, captured by the notion of matching sessions). The idea is that if the adversary is passive it suffices for the private ephemeral keys not to be leaked if the $\INDCPA$ security of the $\KEM$s holds; if the adversary is active, ephemeral $KEM$s do not offer any security by themselves, so we require that the long-term authentication secrets do not fail before the stage accepts. A $\KEMuckle$ session $\pi_i^s$ in stage $t$ is \textit{clean} if:
\begin{itemize}
    \item $\Reveal(i,s,t)$ has not been issued
    \item $\Reveal(j,r,t)$ has not been issued for sessions $\pi_j^r$  matching $\pi_i^s$ at stage $t$.
    \item If $\pi_i^s$ has a matching session $\pi_j^r$; at least one of the following is true:
        \begin{itemize}
            \item If $\pi_i^s$ has the initiator role in stage $t$, $\CompromiseQK(i,s,t)$ has not been issued before $\pi_i^s[t]$ accepts; if it has the responder role, $\CompromiseQK(j,r,t)$ has not been issued before before $\pi_j^r[t]$ accepts
            \item If $\pi_i^s$ has the initiator role in stage $t$, $\CompromiseSK(i,s,t)$ has not been issued; if it has the responder role, $\CompromiseSK(j,r,t)$ has not been issued
        \end{itemize}
    \item If there is no $(j,r,t)$ such that $\pi_j^r$ is an origin session of $\pi_i^s$ in stage $t$, at least one of the following holds:
        \begin{itemize}
            \item $\CorruptQK(i)$ has not been issued before $\pi_j^r$ accepts
            \item $\CompromiseSK(i,s,t)$ has not been issued
        \end{itemize}
\end{itemize}

\paragraph{Resumption stages and symmetric authentication.} $\KEMuckle$ keeps track of the internal variable $\SecState$, which on the first stage of the session is initialized to some public value, and updated as a PRF evaluation of the master secret at the end of the stage. This new value is then fed into a subsequent stage of the protocol.

In subsequent stages of the session (say stage $i$), we can tolerate ephemeral and even long-term secrets failing, provided that we are willing to assume that $\CompromiseSS$ was not issued in stage $i-1$ (arguing by a similar line to that in the proof). On the other hand, if we wish to add this scenario to the cleanness predicate, we might wish to define a version of $\KEMuckle$ that uses more efficient symmetric authentication (or even does not inject further ephemeral secrets), analogously to the resumption sessions of TLS 1.3 \cite{JC:DFGS21}. Such a protocol, more or less, already exists; the original Muckle protocol \cite{PQCRYPTO:DowHanPat20}. As such one can imagine a scenario where we define $\KEMuckle$ as having its first stage consisting of the full protocol, and subsequent stages with $PSK$ authentication. One would in this case have to be a little more careful with the cleanness predicate, and we leave this exercise to further work.

\begin{theorem}
    Let $\prfeval$ be a dual PRF, $\MAC$ be an $ \EUFCMA$ secure MAC, $\KEM_c$ and $\KEM_{pq}$ be $\INDCPA$ secure KEMs, and $\KEM_s$ be an $\INDCCA$ secure KEM. Then the $\KEMuckle$ key exchange protocol is HAKE secure with the cleanness predicate $\clean_{\KEMuckle}$. If the security of $\prfeval$, $\MAC$, $\KEM_{pq}$ and $\KEM_{s}$, or of \qkd{} holds, then so does the security of $\KEMuckle$. 
\end{theorem}

\begin{proof}
We begin by guessing the parameters defining the stage to be tested. 

\begin{description}
    \item[Game A0:] This is the standard HAKE-experiment for an PPT adversary $\advA$. For notational convenience we write
  \[
  \Adv^{\mathsf{HAKE},\clean_{\KEMuckle},C_{5}}_{\advA,\KEMuckle,n_p,n_s,n_t}(\secpar) =: \Adv^{A_0}_{\advA}(\secpar).
  \]
     \item[Game A'0:] In Game A'0, we guess the parameters $(i,s,t)$ defining a stage to be tested, and the intended partner session in that stage. If $\Test(i',s',t')$ is issued with $(i,s,t)\neq (i',s',t')$ we abort the session. The advantage is 
    \[
    \Adv^{A_0}_{\advA}(\secpar) \leq n^2_Pn_Sn_T \Adv^{A'_0}_{\advA}(\secpar)
    \]
\end{description}
Since a $\Test$ query issued on any other session to the one we guessed causes an abort, henceforth we can talk about \textit{the} tested session (at the cost of incurring some polynomial factor in the bound).

We now split into two cases: the case that the tested session does not have an origin session in stage $t$, and the case that it does. We start with the former case.

\paragraph*{Case 1: no origin session.}

Suppose that a session $\pi_i^s$ does not have an origin session in stage $t$. Intuitively, the ephemeral KEMs provide no security here, since our opponent is active and can encapsulate their own choice of secret; the strategy of the proof is to start by replacing the secret encapsulated by the long term KEM, and then proceed by a series of game hops to an unwinnable game. For this replacement to be sound, we require that the tested session's intended partner's long-term secret, or the replacement, is trivially detected. If the long-term secret has been compromised, it suffices for the \qkd{} link to survive.

\paragraph*{Subcase 1a: No $\CorruptQK(i,s,t)$ has been issued.}

The arguments are slightly different, depending on whether the tested session was an initiator or a responder. In the interest of handling the more complex case first, we may for now suppose that the tested session is a responder. In this subcase, replacement of the initiator's static secret $ss_I$ is non-problematic.

\paragraph*{Sub-subcase 1a': $\pi_i^s$ is a responder in stage $t$.}

\begin{description}
  \item[Game A'0:] This is the game $A'0$ described above, with advantage $\Adv^{A'0}_{\advA}(\secpar)$.

  \item[Game A'1:] In Game $A1$ we replace $ss_I$ with a value $\widetilde{ss_I}$ sampled uniformly at random from the output space of $\KEM_s.\Encaps()$. Any subsequent value computed as a function of $ss_I$ is computed instead as a function of $\widetilde{ss_I}$. Consider our adversary $\advA$ given the data defining either game $A'0$ or game $A1$ (notice the syntax of both games is the same); we will relate $\advA$'s advantage in each of these games by an adversary tackling the $\INDCCA$ security game. Let $\advA1$ be such an adversary against the $\INDCCA$ security of $\KEM_s$. Suppose $\advA1$ receives challenge $(\pk^*,\ctxt^*,ss^*)$ in the $\INDCCA$ game. Since the session is assumed to be clean, $\advA$ will not issue a $\CompromiseQK$ query, so we may safely replace the long-term public key of the corrupted partner with $pk^*$. We can also replace the initiator's encapsulation $(c_I,ss_I)$ with the challenge $(\ctxt^*,ss^*)$. Notice, however, that because the tested responder session does not have an origin session, we are not guaranteed that $\ctxt^*$ is delivered. If this happens, let $c'_I$ received by the tested session. The $\INDCCA$ adversary, in this case, queries its decapsulation oracle on $c'_I$ and uses the response as $ss^*$.

  If $ss^*$ was the real secret that $ct^*$ was encapsulated against, the view of $\advA$ is exactly as though it is playing game $A1$; otherwise it is exactly as though it is playing game $A2$. If we use the output of $\advA$ as the guess for $\advA1$ it follows that
  \[
  \Adv^{A'_0}_{\advA}(\secpar)\leq \Adv^{A_1}_{\advA}(\secpar) + \Adv^{\INDCCA}_{\advA1,\KEM_s}(\secpar)
  \]
  
  \item[Game A'2:] In Game $A2$ we replace $\AHS$ with a value $\widetilde{\AHS}$ chosen uniformly at random from the output space of $\prfeval$. Again, any values computed as a function of $\AHS$ are instead computed as a function of $\widetilde{\AHS}$. We employ a similar argument to the above. Since the adversary is active and the \qkd{} link has not been assumed to have survived, we may assume access to $\dHS$. An adversary $\advA2$ trying to win the $\dPRF$ security game can query their PRF oracle on $\dHS$ and run $\advA$ on Game $A1$ with $\AHS$ replaced with the oracle response. If the oracle was indeed a pseudorandom function we have simulated Game A1 to $\advA$; otherwise we have simulated Game A2 to $\advA$, and so
  \[
  \Adv^{A1}_{\advA}(\secpar)\leq \Adv^{A2}_{\advA}(\secpar) + \Adv^{\dPRF}_{\advA2,\prfeval}(\secpar)
  \]

  \item[Game A'3:] In Game $A3$, $\IAHTS$ is replaced with a value sampled uniformly at random, $\widetilde{\IAHTS}$. Any value derived from $\IAHTS$ is instead computed as a function of $\widetilde{\IAHTS}$. An adversary $\advA3$ trying to win the $\dPRF$ security game can query its oracle on $\ell_9\|H_2$ and run $\advA$ on Game $A2$ with $\IAHTS$  replaced with the oracle response. Since in game $A2$, $\IAHTS$ is computed from $\widetilde{AHS}$, which is a uniformly random value not known to the adversary, if the oracle response was indeed from a pseudorandom function we have exactly simulated Game $A3$ to $\advA$; otherwise we have simulated $A3$. It follows that 
  \[
  \Adv^{A2}_{\advA}(\secpar) \leq \Adv^{A3}_{\advA}(\secpar) + \Adv^{\dPRF}_{\advA3,\prfeval}(\secpar)
  \]

  \item[Games A'4, A'5:] We repeat the argument above for the values $\RAHTS$ and $\dAHS$. With $\advA4,\advA5$ PRF adversaries defined similarly to the above one has
  \[
  \Adv^{A3}_{\advA}(\secpar) \leq \Adv^{A4}_{\advA}(\secpar) + \Adv^{\dPRF}_{\advA4,\prfeval}(\secpar)
  \]
  \[
  \Adv^{A4}_{\advA}(\secpar) \leq \Adv^{A5}_{\advA}(\secpar) + \Adv^{\dPRF}_{\advA5,\prfeval}(\secpar)
  \]  

  Indeed, without loss of generality we may suppose that $\advA4$ and $\advA5$ are the same adversary (since by assumption no PPT adversary has better than trivial advantage in the $\dPRF$ game). This observation yields
  \[\Adv^{A3}_{\advA}(\secpar) \leq \Adv^{A5}_{\advA}(\secpar) + 2\Adv^{\dPRF}_{\advA4,\prfeval}(\secpar)\]
  We will implicitly call upon this observation several times.
  
  \item[Game A'6:] In Game $A6$ we replace $\MS$ with $\widetilde{\MS}$, 
  a value sampled uniformly at random from the output space of $\prfeval$. Any values derived from $\MS$ are instead derived from $\widetilde{MS}$.
  The tested session is an initiator session and the adversary is active, so we may assume access to $ss_R$. A PRF adversary $\advA6$ can therefore query its PRF oracle on $\ell_{13}\|ss_R$ and run $\advA$ on game $A5$ with $\MS$ replaced by the oracle response. Since in Game $A5$, $\dAHS$ is a uniformly random value not known to the adversary, if the oracle output was from a pseudorandom function we have simulated Game $A5$ to $\advA$. Otherwise, we have simulated $A6$, yielding
  \[
  \Adv^{A5}_{\advA}(\secpar) \leq \Adv^{A6}_{\advA}(\secpar) + \adv^{\dPRF}_{\advA6,\prfeval}(\secpar)
  \]
  \item[Game A'7-A'11:] Since $\MS$ is now replaced with a uniformly random value, we may sequentially replace $\fk_I,\fk_R,\IATS,\RATS,\SecState$ with uniformly random values. As usual, any derivatives of these values are instead computed as derivatives of the replacement random values. By considering PRF adversaries $\advA7-11$ proceeding according to the strategy outlined above (and recalling our observation in Games A4 and A5), we conclude that
  \[
  \Adv^{A6}_{\advA}(\secpar) \leq \Adv^{A11}_{\advA}(\secpar) + 5\Adv^{\dPRF}_{\advA7,\prfeval}(\secpar)
  \]

  \item[Identical-until-bad] We now define the event `bad' as that whereby the initiator accepts a tag $RF$.  Let $B_0$ denote the event that the `bad' event occurs in Game A11.

  \item[Game A'12:] Game A12 is exactly the same as Game A11, except we abort if the bad event occurs. Since A11 and A12 are identical-until-bad, by \cite[Lemma~2]{bellare2004code}
  we have 
  \[
  |\Adv^{A11}_{\advA}(\secpar)-\Adv^{A12}_{\advA}(\secpar)|< \Pr(B_0)
  \]
  
  Notice that in Game A12, because the adversary can only win the game if the session accepts, but we abort whenever a session accepts because of the bad event, trivially we have $\Adv^{A12}_{\advA}(\secpar)=0$.
  It remains to bound $\Pr(B_0)$.

  \item[Bounding $\Pr(B_0)$:] Recall we are in the subcase where there is at least one $\bm{m_i}$ that is not faithfully delivered by the adversary. Without loss of generality we may suppose $\bm{m_7}$ was not delivered faithfully - because in Game A11, all the secrets have been succesfully replaced by uniformly random values. If $\bm{m_7}$ is faithfully delivered then the adversary has acted passively in a game in which all the secrets are uniformly random values, and so its advantage is $0$.
  
  Consider an adversary $\advA8$ playing the $\EUFCMA$ security game with respect to a symmetric key $fk_I$, that runs $\advA$ playing Game A11 as a subroutine. We answer its oracle queries under $fk_I$. Recalling that $\advA$ controls all traffic; when it produces $IF$, we submit $(H_3, IF)$ as our forgery. We answer $\advA8's$ oracle queries with the key $fk_I$. Since $fk_I$ has been replaced with a uniform random value, the simulation is sound; and by the above we may assume that no honest origin session delivder $\bm{m_7}$. This is therefore a valid forgery, and so the probability that $B_0$ occurs is bounded by the advantage of $\advA8$ in the $\EUFCMA$ security game; that is,
  \[\Pr(B_0)\leq \Adv^{\EUFCMA}_{\advA8}(\secpar)\]

  Descending the chain of inequalities, we conclude that there are PPT adversaries $\advA1,\advA2$ and $\advA3$ such that
  \[\Adv^{A0}_{\advA}(\secpar)\leq 
  \left(
  \Adv^{\INDCCA}_{\advA1,\KEM_s}(\secpar) + 7\Adv^{\dPRF}_{\advA2,\prfeval}(\secpar) + \Adv^{\EUFCMA}_{\advA3,\MAC}(\secpar)
  \right)
  \]

  \paragraph*{Sub-subcase 1a'': $\pi_i^s$ is an initiator in stage $t$.}
  
  Suppose that the tested session is an initiator. We employ the same general argument; but now we note that the replacement of the initiator's secret $ss_I$ is no longer sound. Instead, notice that we can now replace $ss_R$ with a value sampled uniformly at random and again relate the advantage of the two resulting games by the advantage of and $\INDCCA$ adversary. In a version of the game where $ss_R$ is a uniformly random value not known to the adversary we may regard $\prfeval$ as pseudorandom in its first argument, so we can proceed by the same argument as in the initiator case and invoke the $\dPRF$ security of $\prfeval$, and replace $\MS$ with a value sampled uniformly at random. The proof is then the same as in the initiator case, and it follows that there are PPT adversaries $\advA1,\advA2,\advA3$ such that
  \[
  \Adv^{A'0}_{\advA}(\secpar) \leq \Adv^{\INDCCA}_{\advA1,\KEM_s}(\secpar) + 6\Adv^{\dPRF}_{\advA2}(\secpar) + \Adv^{\EUFCMA}_{\advA3,\MAC}(\secpar) 
  \]

  \paragraph*{Subcase 1b: no $\CompromiseSK(i,s,t)$ has been issued.}

  Having handled the subcase in which the long-term authentication secrets were not compromised, we may now consider the subcase in which no $\CompromiseSK(i,s,t)$ query was issued. In this case the proof is the same for a tested initiator and tested responder session. The idea is that if the \qkd{} link was not compromised we can consider $k_q$ to be a uniformly random value, after which we may proceed by similar arguments to the above.

  \item[Game A'0:] This is the same Game A'0 defined above.

  \item[Game B1:] This game is the same as the previous one, except we replace the chaining key $k_2$ with a value $\widetilde{k_2}$ sampled uniformly at random. Any values derived from $k_2$ are instead derived from $\widetilde{k_2}$. As discussed above, we regard $k_q$ as a value distributed uniformly at random, because the \qkd{} link is assumed to have survived. Since we are in the active case a $\dPRF$ adversary $\advB1$ can be assumed to have access to $k_1$; this adversary can therefore query its PRF oracle on $\ell_4\|k_1$, and run $\advA$ on Game $0$ with $k_2$ replaced with the oracle output. If the oracle was a PRF then we have exactly simulated Game $0$; otherwise we have exactly simulated $B1$, yielding

  \[
  \Adv^{A'_0}_{\advA}(\secpar) \leq \Adv^{B_1}_{\advA}(\secpar) + \Adv^{\dPRF}_{\advB1,\prfeval}(\secpar)
  \]

  \item[Game B2:] In this game we replace the chaining key $k_3$ with the a value sampled uniformly at random, say $\tilde{k_3}$. Any values derived from $k_3$ are instead derived from $\tilde{k_3}$. If $\SecState$ is known to the adversary (either because it is the first stage of the session, or because a $\CompromiseSS$ query was issued on the prior stage), a $\dPRF$ adversary $\advB2$ can query its oracle on $\SecState$ and run $\advA$ on Game $B2$ with $k_3$ replaced by the oracle output. Since $\prfeval$ is also a PRF on its first argument in this case, if the oracle output was from a PRF we have exactly simulated Game $B1$; otherwise we have exactly simulated $B2$, whence
  
  \[
  \Adv^{B_1}_{\advA}(\secpar) \leq \Adv^{B_2}_{\advA}(\secpar) + \Adv^{\dPRF}_{\advB2,\prfeval}(\secpar)
  \]

  Notice that if $\SecState$ was not known, the $\dPRF$ adversary can instead query its oracle on a value chosen uniformly at random, and the same argument applies.
  
  \item[Games B3-5:] $\IHTS,\RHTS$ and $\dHS$ are now computed as PRF evaluations on the uniformly random value $\tilde{k_3}$. By considering PRF adversaries $\advB3,\advB4$ and $\advB5$ running $\advA$ on games $B3,B4$ and $B5$ respectively, and arguing as above, one has

  \[
  \Adv^{B2}_{\advA}(\secpar) \leq \Adv^{B5}_{\advA}(\secpar) + 3\Adv^{\dPRF}_{\advB3}(\secpar)
  \]

  \item[Game B6:] In this game we replace $\AHS$ with a uniform random value, $\widetilde{AHS}$. Any values derived from $\AHS$ are instead derived from $\widetilde{AHS}$. Note that we are in the active adversary scenario so we may assume knowledge of $ss_I$. A PRF adversary $\advB6$ can therefore query its oracle on $\ell_9\|ss_I$ and run $\advA$ on Game $B5$ with $\AHS$ replaced by the oracle response. Since in this game $\dHS$ has been replaced with a uniform random value not known to the adversary, if the oracle was a PRF we have exactly simulated Game $B5$ to $\advA$; otherwise we have simulated Game $B6$, giving

  \[
  \Adv^{B5}_{\advA}(\secpar) \leq \Adv^{B6}_{\advA}(\secpar) + \Adv^{\dPRF}_{\advB6,\prfeval}(\secpar)
  \]

  \item[Games B7-9:] In these games $\AHS$ has been replaced by a uniformly random value. Following the usual argument, we define games $B7-B9$ by sequentially replacing $\IAHTS, \RAHTS$ and $\dAHS$, and consider PRF adversaries $\advB7-9$ running $\advA$ on these games. We obtain

  \[
  \Adv^{B6}_{\advA}(\secpar) \leq \Adv^{B9}_{\advA}(\secpar) + 3\Adv^{\dPRF}_{\advB7,\prfeval}(\secpar) 
  \]

  \item[Game B10:] In this game we replace $\MS$ with a uniform random value, say $\widetilde{MS}$. Any values derived from $\MS$ are instead derived from $\widetilde{MS}$. A PRF adversary $\advB8$ can query its PRF oracle on $\ell_13\|ss_R$ (where we may assume knowledge of $ss_R$ since we are we did not require that the long-term authentication mechanism survives) and run $\advA$ on Game $B9$ with $\MS$ replaced by the oracle response. Since $\dAHS$ has already been replaced by a uniformly random value not known to the adversary, if the oracle was a PRF we have exactly simulated Game $B9$ to $\advA$; otherwise, we have simulated Game $B10$. It follows that

  \[
  \Adv^{B9}_{\advA}(\secpar) \leq \Adv^{B10}_{\advA}(\secpar) + \Adv^{\dPRF}_{\advB8,\prfeval}(\secpar)
  \]

  We are now in exactly the same scenario as Game $A7$ and may proceed according to the same argument. We conclude that there are PPT adversaries $\advB1,\advB2$

  \[
  \Adv^{A_0}_{\advA}(\secpar) \leq 15\Adv^{\dPRF}_{\advB1,\prfeval}(\secpar) + \Adv^{\EUFCMA}_{\advB2,\MAC}(\secpar)
  \]
  \end{description}

  \paragraph*{Case 2: origin session exists.}

  We can now move on to the case where the tested stage $\pi_i^s$ has a matching session $\pi_j^r$ in stage $t$. As we have discussed, this can be thought of as the ``passive'' case. Here we only need one of the ephemeral post-quantum secret, the long-term post-quantum secret, the per-stage $\SecState$ (provided $t>1$), or the out-of-band \qkd{} link to survive. 

  Again, all the game-hops here happen from Game A'0 that we defined before, allowing us to argue about a single session being tested, at the cost of a polynomial factor in the final bound.

  In all of these cases we use exactly the same principles as above to show that if a secret is not compromised, we can set up a PRF game such that the advantage in this PRF game bounds the difference in advantage between the original game and a version of the original game whereby the secret is replaced with a uniformly random value. Notice also that in this matching case we do not have to worry about $\MAC$ forgeries. As such we eschew the details of the proof in these cases and present the following conclusions in terms of advantage. 

  If no $\CompromiseSK(i,s,t)$ or $\CompromiseSK(j,r,t)$ has been issued, we can (regardless of the tested stage's role) replace $ss_{pq}$ with a value chosen uniformly at random, $\widetilde{ss_{pq}}$. Since the adversary is passive, this is a sound replacement. Similarly to what we saw in the active case, we can use the advantage in a specially constructed $\INDCPA$ game to bound the difference between the advantage of the original key-indistinguishability game, and that of a variant of this game, in which the ephemeral post-quantum secret is uniformly random. From there we can sequentially replace secrets by relying on the $\dPRF$ security of $\prfeval$. Moreover, since we are in the passive case, we do not have to worry about $\MAC$ forgeries, so we will not need to invoke the `identical-until-bad' type argument of the previous cases. It follows that in this case, for a PPT adversary $\advA$ we have PPT adversaries $\advC1, \advC2$ such that 
  \[
  \Adv^{A_0}_{\advA}(\secpar) \leq \Adv^{\INDCPA}_{\advC1,\KEM_{pq}}(\secpar) + 18\Adv^{\dPRF}_{\advC2,\prfeval}(\secpar)
  \]
  
  If no $\CompromiseSK(i,s,t)$ or $\CompromiseSK(j,r,t)$ have been issued the proof goes through almost exactly the same as in the previous case, except we no longer have to bound the probability of a $\MAC$ forgery (since we are in the passive scenario). In this case we proceed entirely by relying on the $\dPRF$ security of $\prfeval$. In particular, for a PPT adversary $\advA$ we have a PPT adversary $\advD$ such that 

  \[
  \Adv^{A_0}_{\advA}(\secpar) \leq 15\Adv^{\dPRF}_{\advD,\prfeval}(\secpar)
  \]

  Finally, we notice that we can allow everything apart from the long-term authentication secrets to fail, because KEM authentication introduces new secrecy into the key schedule. (Compare this with signature authentication, in which a passively-observed transcript where the signature keys survive but ephemeral keys fail yields full key recovery). If the tested stage is a responder session, for a PPT adversary $\advA$ there are PPT adversaries $\advE1,\advE2$ such that
  \[
  \Adv^{A_0}_{\advA}(\secpar) \leq \Adv^{\INDCPA}_{\advE1,\KEM_s}(\secpar) + 9\Adv^{\dPRF}_{\advE2,\prfeval}(\secpar)
  \]

  If the tested stage was a responder the bound is
  \[
  \Adv^{A_0}_{\advA}(\secpar) \leq \Adv^{\INDCPA}_{\advE1,\KEM_s}(\secpar) + 5\Adv^{\dPRF}_{\advE2,\prfeval}(\secpar)
  \]

  Notice that in comparison with the active case, we now only need the $\INDCPA$ security of $\KEM_s$.
\end{proof}

\begin{remark}
    In contrast to \KEMTLS \cite{CCS:SchSteWig20}, we do not require the ephemeral $\KEM$s to have $\INDONECCA$ security (and instead only require $\INDCPA$ security). This is because the proof of \KEMTLS and our proof split into two cases slightly differently. In the proof of \KEMTLS, there is a case whereby an honest matching session exists, and a case in which an honest matching session does not exist, but in both scenarios the adversary is allowed to behave actively. There is no guarantee that the adversary faithfully delivers the encapsulation under the ephemeral $\KEM$, so for the security proof to go through, we may require a decapsulation query in order to replace the ephemeral secret in a valid fashion. In contrast, our proof divides into a case where the session we partner with has a corresponding session with a matching transcript, which implies that the adversary has acted passively. As such, we know that the encapsulation of the ephemeral secret has been delivered faithfully, and so we do not need to rely on a decapsulation query. In the case where there is no session with a matching transcript, the adversary must have behaved actively -- but in this case, the ephemeral KEMs do not offer any security anyway, since an active adversary can safely encapsulate their own secret against the initiator's public key (or send their own ephemeral public keys to a responder), and we rely on the authentication mechanism. In summary, in either of our cases, there is no need to upgrade the $\INDCPA$ security of the ephemeral $\KEM$ to $\INDONECCA$ security. Note that this is in line with the results from Huguenin-Dumittan and Vaudenay \cite{EC:HugVau22} where the authors show that CPA-secure ephemeral KEMs are essentially sufficient for TLS 1.3. Our work shows that such guarantees also hold for \MuckleSharp in the HAKE framework.
\end{remark}

\ifdefined\full
\section{Implementation and Evaluation}

We implemented a prototype of the $\KEMuckle$ protocol in Python using bindings\footnote{\url{https://github.com/open-quantum-safe/liboqs-python}} of \texttt{liboqs}~\cite{SAC:SteMos16} for the support of post-quantum signature schemes and KEMs as well as the \texttt{cryptography}\footnote{\url{https://pypi.org/project/cryptography/}} module for all classically-secure schemes. The initiator (and the responder) were executed on a notebook running Windows~10 with an Intel~i5 2.60GHz CPU and 8~GB of RAM reflecting a client (and server) view. 

\paragraph{About the implemented architecture.} The architecture of an application integrating the $\KEMuckle$ protocol is depicted in \Cref{fig:Muckle_Architecture}. The protocol is implemented in the application layer of a \qkd{} network. The quantum key material is fetched by the key managements services (KMSs) \cite{DBLP:conf/IEEEares/JamesLRT23,DBLP:conf/icton/MartinBOBVSSACSSEDRPL23,DBLP:journals/csur/MehicNRMPAMSPPV20} from the \qkd{} devices and is then provided via the ETSI GS QKD 014~\cite{ETSI014} interface to the application layer. With this interface, the initiator-side of the application obtains a key from the KMS together with an associated ID. For the responder to request the same key, the initiator needs to communicate the key ID to the responder which then requests the key by providing the key ID. Hence, this ID needs to be included in the initial message of the key exchange but no other information needs to be transmitted for the initiator and responder to use the ETSI GS QKD 014 interfaces.

To go into more detail regarding the additional communication cost, we want to note that our protocol requires one more message compared to TLS 1.3 from the initiator (i.e., the client in the TLS setting) to the responder (i.e., the server). The initiator/client is the first party to be able to send encrypted data instead of the responder/server. The server being able to send the first encrypted message is only useful for certain protocols such as SMTP, but not for HTTPS where the client sends the message. For client-initiated protocols, the additional message has no negative impact. The concrete trade-offs depend, of course, on the concrete choice of primitives. 

\begin{figure}[t]
  \centering
\begin{tikzpicture}
  \node(QKD1) [QKDmodule]{QKD};
  \node(QKD2) [QKDmodule, right of=QKD1, xshift=2cm]{QKD};
  \draw[arrowtwo] (QKD1) -- node[above] {\scriptsize QKD link} (QKD2);
  \node(QKD3) [QKDmodule, right of=QKD2, xshift=0cm]{QKD};
  \node(QKD4) [QKDmodule, right of=QKD3, xshift=2cm]{QKD};
  \draw[arrowtwo] (QKD3) -- node[above] {\scriptsize QKD link} (QKD4);

  \node(KM1) [KM, above of=QKD1, yshift= 0.3cm, xshift=-0.5cm]{KMS};
  \node(KM2) [KM, above of=QKD2, yshift= 0.3cm, xshift=0.5cm]{KMS};
  \draw[arrowtwo] (KM1) -- node[above] {\scriptsize KMS link} (KM2);
  \node(KM3) [KM, above of=QKD4, yshift= 0.3cm, xshift=0.5cm]{KMS};
  \draw[arrowtwo] (KM2) -- node[above] {\scriptsize KMS link} (KM3);

  \draw [arrow] (QKD1) -- (QKD1 |- KM1.south);
  \draw [arrow] (QKD2) -- (QKD2 |- KM2.south);
  \draw [arrow] (QKD3) -- (QKD3 |- KM2.south);
  \draw [arrow] (QKD4) -- (QKD4 |- KM3.south);

  \node(app1) [app, above of=KM1, yshift=1.8cm] {Initiator Appl.};
  \node(app2) [app, above of=KM3, yshift=1.8cm] {Responder Appl.};

  \node(muckle) [rectangle, draw=blue!80, dashed, above of=KM2, yshift=1.6cm, align=center, minimum height=2cm, minimum width=11cm] {~ \\ ~ \\ ~ \\ \color{blue!80}{$\KEMuckle$}};

  \draw [arrow] (KM1) -- node [left, yshift=0.6cm] {\footnotesize $k_q$} (KM1 |- app1.south);
  \draw [arrow] (KM3) -- node [right, yshift=0.6cm] {\footnotesize $k_q$} (KM3 |- app2.south);

  \draw [arrowtwo] (0.5,3.9) -- node[above] {\footnotesize $k_c$} (6.35,3.9);
  \draw [arrowtwo] (0.5,4.4) -- node[above] {\footnotesize $k_{pq}$} (6.35,4.4);

  \draw[dashed, black!60] (-2,2.1) -- node [above] {\footnotesize ETSI GS QKD 014} (9,2.1);

  \node(node1) [node, fit=(QKD1) (KM1)]{};
  \node(node2) [node, fit=(QKD2) (KM2) (QKD3)]{};
  \node(node3) [node, fit=(QKD4) (KM3)]{};
  \node(text1) [desc, below of=node1, yshift=-0.3cm]{Initiator node};
  \node(text2) [desc, below of=node2, yshift=-0.3cm]{Trusted node};
  \node(text3) [desc, below of=node3, yshift=-0.3cm]{Responder node};
\end{tikzpicture}
\caption{Architecture of the $\KEMuckle$ implementation. $k_q$ represents the \qkd{} key, $k_{pq}$ the key obtained from a post-quantum secure KEM, and $k_c$ the key obtained from a classically secure KEM. KMS represents the key management system. Note that the protocol supports any number of trusted nodes in between. The architecture is adapted from  $\SigMuckle$ from~\cite{PQCRYPTO:BruRamStr23}.}
\label{fig:Muckle_Architecture}
\end{figure}

\paragraph{Further requirements and assumptions.} We evaluated the implementation of $\KEMuckle$ with a simulated link with a single post-quantum secure certificate authority. We used simulators instead of actual \qkd{} devices to ensure that the limited key rate does not influence the practical results (a \qkd{} key has 32 bytes). The benchmark was performed using mutual authentication and hence both parties authenticated themselves using certificates. 
The certificates contained both post-quantum and classical long-term public keys. They were also authenticated in a two-layer certificate hierarchy, i.e., with one root CA and one intermediate CA. The CA signed the certificates using both a post-quantum signature scheme (ML-DSA-87) and a classical signature scheme, namely EdDSA~\cite{CHES:BDLSY11}.





\paragraph{Concrete analysis.} The performance is evaluated based on two critical metrics:  bandwidth and CPU cycles. The following analysis compares these metrics for different schemes as depicted in~\Cref{tab:kemuckle}.

\begin{table}[htbp]
\centering
\caption{Bandwidth and CPU usage in kilobytes and giga cycles, respectively, for \KEMuckle over 5 runs. We use ML-DSA-87 for the post-quantum signature scheme of the long-term KEM keys. (Unfortunately, we were not able to measure cycles for FrodoKEM-976-SHAKE due to an error in a software library.)}
\label{tab:kemuckle}
\resizebox{\textwidth}{!}{
\begin{tabular}{llcc}
\toprule
Ephemeral Keys & Long-Term Keys & Cycles (G) & Data (KB) \\
\midrule
\multirow{1}{*}{\shortstack{ML-KEM-512, ECDH (X25519), QKD}} 
  & ML-KEM-512 & 1.6 & 29.2 \\
\midrule
\multirow{1}{*}{\shortstack{ML-KEM-768, ECDH (X25519), QKD}}  
  & ML-KEM-768 & 1.7 & 31.3 \\
\midrule
\multirow{1}{*}{\shortstack{ML-KEM-1024, ECDH (X25519), QKD}} 
  & ML-KEM-1024 & 1.7 & 33.9 \\
\midrule
\multirow{1}{*}{\shortstack{HQC-128, ECDH (X25519), QKD}} 
  & HQC-128 & 2.9 & 44.5 \\
\midrule
\multirow{1}{*}{\shortstack{HQC-256, ECDH (X25519), QKD}} 
  & HQC-256 & 6.0 & 89.4 \\
\midrule
\multirow{1}{*}{\shortstack{FrodoKEM-640-SHAKE, ECDH (X25519), QKD}} 
  & FrodoKEM-640-SHAKE & 1.8 & 82.5 \\
\midrule
\multirow{1}{*}{\shortstack{FrodoKEM-976-SHAKE, ECDH (X25519), QKD}}
  & FrodoKEM-976-SHAKE & - & 118.6 \\
\midrule
\multirow{1}{*}{\shortstack{FrodoKEM-1344-SHAKE, ECDH (X25519), QKD}}
  & FrodoKEM-1344-SHAKE& 2.4 & 153.9 \\
\bottomrule
\end{tabular}
}
\end{table}

The highest bandwidth is observed using FrodoKEM-1344-SHAKE for long-term keys at 153.9 KB, followed by FrodoKEM-976-SHAKE at 118.6 KB. ML-KEM-512 and ML-KEM-768 exhibit the lowest bandwidth at 29.2 KB and 31.3 KB, respectively. The schemes HQC-128 and HQC-256 have moderate bandwidth values of 44.5 KB and 89.4 KB, respectively.

The computational demand in terms of CPU cycles is as follows over 5 runs. HQC-256 for long-term keys requires the highest CPU cycles at 6.0 G cycles. HQC-128 and FrodoKEM-1344-SHAKE have moderate values of 2.9 G cycles and 2.4 G cycles, respectively. ML-KEM-512, ML-KEM-768 and ML-KEM-1024 are the most efficient with CPU cycles of 1.6, 1.7 and 1.7 G cycles, respectively, while FrodKEM-640-SHAKE has 1.8 G cycles.

\paragraph{Conclusion.} Using Schemes such as ML-KEM-768 and ML-KEM-1024 in \KEMuckle as long-term keys for authentication are efficient in both bandwidth and CPU cycles while also guaranteeing a fairly suitable level of security. FrodoKEM-1344-SHAKE offers higher security but at the cost of increased bandwidth and computational requirements. Nevertheless, the choice of the scheme depends on the specific requirements of bandwidth efficiency versus computational performance in the specific post-quantum cryptographic applications. 



\begin{remark}
Generally note, though, the major bottleneck is the key rate of the underlying \qkd{} network. With key rates in the range of a few kilobits per second,\footnote{For example see \url{https://www.thinkquantum.com/quky/} for practical key rates of a device supplier and \cite{DBLP:conf/IEEEares/DoberlEGKKR23}.} the cost of all involved cryptographic primitives and the cost of the transmission is insignificant. Nevertheless, since \qkd{} is modelled in a black-box way in this work, one can expect the efficiency gains of the post-quantum primitives to be increasingly germane as the relevant technology improves.

Moreover, note that $\KEMuckle$ can also benefit from the pre-distributed certificates~\cite{ESORICS:SchSteWig21}. As the certificates are a major contributor to the size of the handshake messages,\footnote{A single ML-DSA-87 signature is 4627 bytes large, so the two signatures required for a setting with a root CA and an intermediate CA account for almost 9 KB.} pre-distribution of the certificates has a significant impact on the communication overhead. Considering the bandwidth cost measured during a single protocol run (c.f. \Cref{tab:kemuckle}), the certificates are the single largest contributor to the required bandwidth. Considering current efforts to standardize Encrypted Client Hello for TLS~\cite{encryted-client-hello}, out-of-band communication of certificates will become a practical possibility for TLS and other protocols to save communication overhead. This proposal foresees the distribution of certificates via DNS records which will cause other challenges to be tackled with respect to post-quantum signatures, e.g., to the limited sizes of DNS records~\cite{DBLP:conf/space/RawatJ23}.
\end{remark}

\fi

\section{Conclusion and Outlook}

We propose a novel variant of hybrid authenticated key exchange protocols, in which authentication solely relies on the availability of post-quantum KEMs, instead of digital signatures or pre-shared keys; and prove its security in the HAKE framework. \ifdefined\full\else Inspired by KEMTLS~\cite{CCS:SchSteWig20}, our novel HAKE variant significantly improves upon efficiency compared to prior work. \fi In terms of future direction, we echo the sentiment in the conclusion of \cite{PQCRYPTO:DowHanPat20}; that is, it would be interesting to integrate the physical models used to prove \qkd{} into our framework, particularly with respect to the limitations set out in Section~\ref{sec:modelling-qkd}.

\ifdefined\full
\paragraph{Acknowledgements.} This work received funding from the European Union's Horizon Europe research and innovation programme under agreement number 101114043 (``QSNP'') and from the Digital Europe Program under grant agreement number 101091642 (``QCI-CAT''). Authors CB and LP conducted this work with the support of ONR Grant 62909-24-1-2002. LP also thanks Google for supporting academic post-quantum research. Authors CB, CS, and LP acknowledge support of the Institut Henri Poincaré (UAR 839 CNRS-Sorbonne Université), and LabEx CARMIN (ANR-10-LABX-59-01).
\else
\fi

\ifdefined\full
\bibliographystyle{alphaurl}
\else
\fi
\bibliography{bib,cryptobib/abbrev3,cryptobib/crypto}


\end{document}